\documentclass[lettersize,onecolumn]{IEEEtran}
\usepackage{amsmath,amssymb,amsthm,mathtools,bm}
\usepackage{aliascnt}
\usepackage{enumitem}
\usepackage{microtype}
\usepackage{xcolor}
\usepackage{tikz}
\usetikzlibrary{arrows.meta,positioning,calc,fit,decorations.pathreplacing}
\usepackage{hyperref}
\usepackage[nameinlink,noabbrev]{cleveref}
\hypersetup{hidelinks,pdftitle={A Self-Coupling Gaussian Extremal Inequality for Information-Constrained Optimal Transport with Applications to Rate-Distortion-Perception},pdfauthor={Jiayang Zou, Jun Chen, Ayfer Ozgur, Jia Wang, Tsachy Weissman}}
\allowdisplaybreaks[2]
\newtheorem{theorem}{Theorem}
\newaliascnt{lemma}{theorem}
\newtheorem{lemma}[lemma]{Lemma}
\aliascntresetthe{lemma}
\newaliascnt{proposition}{theorem}
\newtheorem{proposition}[proposition]{Proposition}
\aliascntresetthe{proposition}
\newaliascnt{corollary}{theorem}
\newtheorem{corollary}[corollary]{Corollary}
\aliascntresetthe{corollary}
\theoremstyle{definition}
\newaliascnt{definition}{theorem}

\aliascntresetthe{definition}
\theoremstyle{definition} 
\newaliascnt{remark}{theorem}
\newtheorem{remark}[remark]{Remark}
\aliascntresetthe{remark}
\crefname{theorem}{Theorem}{Theorems}
\crefname{lemma}{Lemma}{Lemmas}
\crefname{proposition}{Proposition}{Propositions}
\crefname{corollary}{Corollary}{Corollaries}
\crefname{definition}{Definition}{Definitions}
\crefname{remark}{Remark}{Remarks}
\newcommand{\R}{\mathbb R}
\newcommand{\E}{\mathbb E}
\newcommand{\dd}{\mathrm{d}}
\newcommand{\mcN}{\mathcal N}
\newcommand{\mcP}{\mathcal P}
\newcommand{\Tr}{\operatorname{Tr}}
\newcommand{\Sym}{\operatorname{Sym}}
\newcommand{\Cov}{\operatorname{Cov}}
\newcommand{\Var}{\operatorname{Var}}
\newcommand{\Law}{\operatorname{Law}}
\newcommand{\diag}{\operatorname{diag}}

\newcommand{\pos}[1]{\left[#1\right]_+}
\newcommand{\norm}[1]{\left\lVert#1\right\rVert}
\newcommand{\ip}[2]{\left\langle#1,#2\right\rangle}
\newcommand{\KL}[2]{\mathrm{KL}\!\left({#1}\,\middle\|\,{#2}\right)}

\newcommand{\cA}{\mathcal A}
\newcommand{\cE}{\mathcal E}

\setlist[enumerate]{leftmargin=*,itemsep=2pt}

\usepackage{amsmath,amsfonts}
\usepackage{algorithmic}
\usepackage{algorithm}
\usepackage{array}
\usepackage[caption=false,font=normalsize,labelfont=sf,textfont=sf]{subfig}
\usepackage{textcomp}
\usepackage{stfloats}
\usepackage{url}
\usepackage{verbatim}
\usepackage{graphicx}
\usepackage{cite}
\begin{document}

\title{A Self-Coupling Gaussian Extremal Inequality for Information-Constrained Optimal Transport with \\ Applications to Rate-Distortion-Perception}

\author{%
Jiayang~Zou,
Jun~Chen,~\IEEEmembership{Senior Member,~IEEE,}
Ayfer~\"Ozg\"ur,~\IEEEmembership{Senior Member,~IEEE,}
Jia~Wang
%
\thanks{%
J. Zou and A. \"Ozg\"ur are with the Department
of Electrical Engineering, Stanford University,
Stanford, CA 94305, USA
(e-mail: jyangzou@stanford.edu; aozgur@stanford.edu).%
}%
\thanks{%
J. Chen is with the Department of Electrical and Computer
Engineering, McMaster University, Hamilton, ON L8S 4L7, Canada
(e-mail: chenjun@mcmaster.ca).%
}%
\thanks{%
J. Wang is with the Institute of Image
Communication and Network Engineering, Shanghai Jiao Tong
University, Shanghai 200240, China
(e-mail: jiawang@sjtu.edu.cn).%
}%
}


\maketitle

\begin{abstract}
How much mutual information is needed to achieve a target correlation  between two copies of a prescribed distribution? We prove a sharp Gaussian extremal inequality for this self-coupling problem: among distributions with a fixed mean and covariance, the Gaussian has the greatest minimum information cost for meeting a positive semidefinite cross-covariance target. The inequality is strict for non-Gaussian distributions when the target is positive definite.
The proof uses a midpoint gluing construction. Two Ornstein--Uhlenbeck half-bridges, each started from the prescribed endpoint law, meet at a common intermediate distribution and are joined by conditional gluing. Changing the Gaussian reference then turns the resulting relative-entropy bound into the desired extremal inequality.
As a first application, scaling one endpoint of the self-coupling establishes Gaussian extremality for rate-distortion-perception with unlimited common randomness, under quadratic distortion and squared Wasserstein-2 perception measures. We also establish an analogous result for finite common randomness, where the problem involves two mutual information constraints associated with an auxiliary random variable. In that setting, we retain an intermediate observation on one half-bridge and balance the two information budgets simultaneously using the Poincar\'{e}-Miranda theorem.
\end{abstract}

\begin{IEEEkeywords}
  Gaussian extremality, information-constrained optimal transport, Ornstein--Uhlenbeck, rate-distortion-perception theory, self-coupling. 
\end{IEEEkeywords}

\section{Introduction}\label{sec:introduction}

Information-constrained optimal transport asks how much two prescribed
distributions can be brought into agreement using a limited amount of
statistical dependence. Its basic problem, as defined in \cite{BaiWuOzgur2023}, is
\begin{equation}\label{eq:icot-introduction}
 \inf_{\substack{\pi\in\Pi(\mu,\nu)\\ I_\pi(X;Y)\le R}}
 \E_\pi c(X,Y),
\end{equation}
where $\Pi(\mu,\nu)$ is the set of  couplings of $\mu$ and $\nu$,
$c$ is a transport cost, and
$I_\pi(X;Y)=\KL{\pi}{\mu\otimes\nu}$ is the mutual information under $\pi$, with
$\mathrm{KL}$ denoting Kullback--Leibler divergence (also known as relative entropy). Without the mutual information constraint, \eqref{eq:icot-introduction} reduces to the classical Kantorovich transport problem.

The role of the information constraint is especially clear in the self-coupling case $\nu=\mu$. Under quadratic cost, the identity coupling gives zero transport cost, whereas a finite information budget limits how closely the two copies can track each other. If their common mean is $m$ and covariance is $\Sigma$, then
\begin{equation}\label{eq:self-cost-introduction}
 \E\norm{X-Y}^2
 =2\Tr\Sigma-2\E\langle X-m,Y-m\rangle.
\end{equation}
Minimizing transport cost thus amounts to maximizing self-correlation: how much correlation can a given mutual information budget support? We establish a sharp Gaussian benchmark for this question.

The information-theoretic roots of \eqref{eq:icot-introduction} go back to
Shannon's source-coding theorems and the rate-distortion formalism
\cite{Shannon1948,Shannon1959,Berger1971,CoverThomas2006}. In classical
rate-distortion theory, for a source $X\sim\mu$, one minimizes the expected
distortion $\E c(X,\widehat X)$ over conditional distributions
$P_{\widehat X|X}$ subject to a mutual information constraint
$I(X;\widehat X)\le R$; equivalently, the source marginal $\mu$ is fixed
while the reproduction marginal is left free to be optimized. If the
reproduction marginal is additionally prescribed to be $\nu$, the resulting
problem is precisely of the form \eqref{eq:icot-introduction}. Such
output-constrained source-coding problems were studied in
\cite{SaldiLinderYuksel2015Randomized}.

The viewpoint of interpreting \eqref{eq:icot-introduction} as an information-constrained optimal transport problem, and its connection to transport--entropy theory, was introduced in \cite{BaiWuOzgur2023}. In particular, \cite{BaiWuOzgur2023} showed that this viewpoint yields a strengthening and information-constrained generalization of Talagrand's Gaussian transportation inequality, which bounds quadratic transport cost to a Gaussian reference in terms of relative entropy \cite{Talagrand1996,BobkovGotze1999,Villani2009}. In the same work, the strengthened inequality was applied to derive a relay-channel converse, recovering the resolution of Cover's relay-channel problem in the Gaussian case \cite{WuBarnesOzgur2019}.
Related information-constrained inequalities have also appeared in
soft-minoration and rate-distortion integral bounds
\cite{Liu2023SoftMinoration,Liu2026RateDistIntegral}. A Lagrangian relaxation
of the mutual-information constraint in \eqref{eq:icot-introduction} leads to
entropy-regularized optimal transport and its computational developments,
including Sinkhorn regularization \cite{Cuturi2013,PeyreCuturi2019}.

In this work, we study the information-constrained optimal transport problem
in the special case where the two prescribed marginals are the same. Under
quadratic cost, this self-coupling problem amounts to asking how much
correlation can be created between two copies of a given distribution using a
limited amount of mutual information. For a scalar law $\mu$ with mean $m$
and variance $\sigma^2>0$, our main result takes the form
\begin{equation}\label{eq:scalar-result-introduction}
\sup_{\substack{\pi\in\Pi(\mu,\mu)\\I_\pi(X;Y)\le R}}
 \E_\pi[(X-m)(Y-m)]
 \ge \sigma^2\sqrt{1-e^{-2R}}.
\end{equation}
The Gaussian law attains equality, whereas every non-Gaussian scalar law
satisfies a strict inequality for $R>0$. Thus, among distributions with a
given variance, the Gaussian requires the most mutual information to achieve
a prescribed level of self-correlation. More generally, for a multivariate law with covariance matrix
$\Sigma\succ0$, our main theorem shows that the minimum mutual information
required to achieve a prescribed positive semidefinite cross-covariance
target $0\preceq S\prec\Sigma$ is no larger than that required for a Gaussian
law with the same covariance matrix. The comparison is strict for
non-Gaussian laws whenever $S\succ0$.

This result belongs to a broader family of Gaussian extremality phenomena in
information theory. Under covariance constraints, Gaussian distributions
maximize entropy and play extremal roles in rate-distortion theory and many
information inequalities; related Gaussian comparison principles have been
developed more recently in \cite{Courtade2018,ArasCourtade2022}. The
self-coupling problem considered here, however, poses a different type of
extremal question: the marginal distribution itself is prescribed, and the
optimization is over the dependence between two copies of that law. Even
when the common marginal is Gaussian, admissible couplings need not be
jointly Gaussian, so establishing a sharp Gaussian comparison requires
controlling the full class of couplings.

The appearance of correlation invites a comparison with
Hirschfeld--Gebelein--R\'enyi maximal correlation
\cite{Hirschfeld1935,Gebelein1941,Renyi1959}. Maximal correlation optimizes over square-integrable transforms of a fixed pair of random variables. In
\eqref{eq:scalar-result-introduction}, by contrast, we optimize over joint laws with prescribed marginals. One might be inclined to prove \eqref{eq:scalar-result-introduction} by starting with a correlated Gaussian pair that attains equality in the bound and pushing both coordinates through a nonlinear transport map to the target marginal $\mu$. The resulting coupling, however, may lose correlation and thus fail to establish the desired bound. Indeed, the Hermite expansion and Mehler formula show that nonlinear components decay faster than linear ones \cite{Janson1997}. We therefore construct the dependence directly from the prescribed marginal law $\mu$.



Our proof begins by relaxing the correlation constraint via convex duality. Changing to a correlated Gaussian reference then absorbs the quadratic term in the objective, reducing the comparison to an upper bound on the minimum relative entropy over self-couplings. To obtain this bound, we send each endpoint law through the same Ornstein--Uhlenbeck (OU) smoothing kernel. The resulting half-bridges share an intermediate marginal and can therefore be joined by sampling their endpoints independently conditional on a common midpoint. This construction preserves both endpoint laws. The relative entropy of the path law is the sum of the two endpoint relative entropies minus the relative entropy of the shared midpoint. Projecting onto the endpoints then gives the coupling bound needed for the dual comparison. The midpoint relative entropy also provides the gain that makes the inequality strict for non-Gaussian laws when the cross-covariance target is positive definite.

The self-coupling inequality has a direct application to rate-distortion-perception (RDP) theory \cite{chen2026overview}, initiated by Blau and Michaeli through the introduction of  a distributional perception constraint into the rate-distortion (RD) framework \cite{BlauMichaeli2019}.
Following Freirich et al. \cite{Freirich2021},  in this work we consider quadratic distortion and squared  Wasserstein-2  perception constraints. Under perfect perception, the reconstruction must have the source law, so the self-coupling inequality provides a starting point. Allowing the two laws to differ permits a reduction in distortion by shrinking the reconstruction toward the  source mean. This displacement controls the perception loss, while the self-correlation bound controls distortion; the operation requires no additional mutual information. Choosing the amount of shrinkage appropriately gives a simple proof that, among sources with the same mean and covariance, the Gaussian has the largest minimum achievable distortion. The comparison is strict for non-Gaussian sources whenever the Gaussian optimal rate allocation assigns positive rate to every covariance eigendirection.

A further question arises when the randomness shared by the encoder and decoder is limited
\cite{SaldiLinderYuksel2015CommonRandomness,Wagner2022,XLCZ2025,XLCYZ2026}. The single-letter characterization of the relevant RDP limit  in \cite{QuChenYuXu2025} involves two couplings linking the respective endpoints to a common auxiliary random variable, with a separate mutual information constraint on each coupling.
We use a variational identity to recast this characterization in the form of a weighted sum of quadratic errors, bringing the comparison back to moments of a self-coupling. To accommodate the unequal information budgets, we retain an intermediate observation on one half-bridge. Simultaneous rate balancing, justified by the Poincar\'{e}--Miranda theorem, selects observation strengths that satisfy both information constraints and the required moment bounds
\cite{Miranda1940,Kulpa1997}.
This proves Gaussian extremality for every finite common-randomness budget, with strictness for non-Gaussian sources under the same positive rate allocation condition on the Gaussian optimizer. 


The paper is organized as follows.
\Cref{sec:self-coupling} proves the self-coupling
inequality and discusses its consequences.
\Cref{sec:rdp} applies the inequality to establish Gaussian extremality for quadratic Wasserstein RDP  with unlimited common randomness. 
\Cref{sec:rdpc-formulation} treats the more challenging setting of limited common randomness.
\Cref{sec:conclusion} concludes the paper.
The appendices provide supporting arguments on duality, approximation, Gaussian optimization, continuity, and variational identities.

\section{The self-coupling inequality}
\label{sec:self-coupling}
\subsection{Statement and consequences}

 Let $\mcP_2(\R^d)$ denote the set of probability laws on $\R^d$ with finite second moment.
The relative entropy  of $P$ with respect to $Q$  is $\KL{P}{Q}=\int\log(\dd P/\dd Q)\,\dd P$ when $P\ll Q$,
and is $+\infty$ otherwise. For a coupling $\pi\in\Pi(\mu,\nu)$, the mutual information is $I_\pi(X;Y)=\mathrm{KL}(\pi\|\mu\otimes\nu)$.
We also write $(X,Y)\in\Pi(\mu,\nu)$ when the joint law of $(X,Y)$ belongs to this
coupling class. Throughout the paper, logarithms are natural.

Fix a probability law $\mu$ on $\R^d$ with mean $m$ and covariance $\Sigma$. A self-coupling is a pair $(X,Y)$ whose marginals both equal $\mu$. The choice lies in their dependence: independence requires no mutual information, while bringing the two copies into closer agreement generally requires more. For quadratic transport cost, this tradeoff is particularly clear because
\begin{equation}\label{eq:self-cost-correlation}
 \E\norm{X-Y}^2
 =2\Tr\Sigma-2\E\ip{X-m}{Y-m}.
\end{equation}
Thus, minimizing transport cost under an information budget amounts to
maximizing the correlation between two copies of the same law.

For a jointly Gaussian scalar pair with marginal variance $\sigma^2$ and correlation coefficient $s\in[0,1)$, the mutual information is $-\frac12\log(1-s^2)$. Consequently, the largest centered covariance achievable at information budget $R$ is $\sigma^2\rho(R)$, where
\begin{equation}\label{eq:rho-definition}
 \rho(r):=\sqrt{1-e^{-2r}},\qquad r\ge0.
\end{equation}
Our first result shows that every law with the same covariance can attain or exceed this Gaussian benchmark. Its matrix form accommodates different correlation requirements in different directions. We write $\Sym H=(H+H^\top)/2$ and use $\preceq$ for the Loewner order.

\begin{theorem}[The self-coupling inequality]
\label{thm:matrix-self-coupling}
Let $\mu\in \mcP_2(\R^d)$ have mean $m$ and covariance $\Sigma\succ0$.  For every symmetric
$S$ satisfying $0\preceq S\prec\Sigma$, define
$\tilde{S}=\Sigma^{-1/2}S\Sigma^{-1/2}$.  Then
\begin{equation}\label{eq:matrix-coupling-result}
 \min_{\substack{(X,Y)\in\Pi(\mu,\mu)\\
       \Sym\E[(X-m)(Y-m)^\top]\succeq S}}
 I(X;Y)
 \le -\frac12\log\det(I_d-\tilde{S}^2).
\end{equation}
If $\mu=\mcN(m,\Sigma)$, equality holds, and a jointly Gaussian pair with cross-covariance $S$ attains the minimum. If $\mu$ is non-Gaussian, the inequality is strict whenever $0\prec S\prec\Sigma$.
\end{theorem}

The proof of Theorem \ref{thm:matrix-self-coupling} uses a midpoint gluing construction. We pass both endpoint laws through the same Gaussian smoothing kernel, bringing them to a common intermediate law. The resulting half-bridges can then be joined by sampling their endpoints conditionally on the shared midpoint. The entropy of the midpoint provides a gain in the comparison and yields strictness for non-Gaussian laws when the target is positive definite. The details can be found in \Cref{sec:midpoint-proof} and \Cref{subsec:midpoint-entropy}.

\begin{corollary}[Scalar self-correlation]
\label{cor:scalar-self-correlation}
\label{cor:strict-scalar-self-coupling}
Let $\mu$ be a law on $\R$ with mean $m$ and variance $\sigma^2>0$.
For every finite $R\ge0$, there exists $(X,Y)\in\Pi(\mu,\mu)$ such that
\begin{equation}\label{eq:scalar-self-correlation}
 I(X;Y)\le R,
 \qquad \E[(X-m)(Y-m)]\ge\sigma^2\rho(R).
\end{equation}
The maximum achievable covariance equals $\sigma^2\rho(R)$ for Gaussian
$\mu$.  If $R>0$ and $\mu$ is non-Gaussian, there exists a self-coupling satisfying
\begin{equation}\label{eq:strict-scalar-self-correlation}
 I(X;Y)\le R,
 \qquad \E[(X-m)(Y-m)]>\sigma^2\rho(R).
\end{equation}
\end{corollary}

Corollary \ref{cor:scalar-self-correlation}  implies that
 among laws with the same variance, the Gaussian requires the most mutual information to achieve a prescribed positive self-correlation. Its proof is  given at the end of
\Cref{subsec:midpoint-entropy}.


Two immediate consequences of Theorem \ref{thm:matrix-self-coupling} follow. First, one can prescribe the full covariance matrix of the transport displacement.

\begin{corollary}[Matrix distortion]
\label{cor:matrix-distortion}
Under the assumptions of \Cref{thm:matrix-self-coupling}, let
$0\prec A\preceq2\Sigma$ and define $\bar A=\Sigma^{-1/2}A\Sigma^{-1/2}$.
Then
\begin{equation}\label{eq:matrix-distortion}
 \min_{\substack{(X,Y)\in\Pi(\mu,\mu)\\
                 \E[(X-Y)(X-Y)^\top]\preceq A}}
 I(X;Y)
 \le -\frac12\log\det\left(\bar A-\frac14\bar A^2\right).
\end{equation}
Equality holds when $\mu$ is Gaussian with covariance $\Sigma$.
\end{corollary}

\begin{proof}
Since both marginals have covariance $\Sigma$, 
\[
 \E[(X-Y)(X-Y)^\top]
 =2\bigl(\Sigma-\Sym\E[(X-m)(Y-m)^\top]\bigr).
\]
Apply \Cref{thm:matrix-self-coupling} with $S=\Sigma-A/2$, and note that
$I_d-(I_d-\bar A/2)^2=\bar A-\bar A^2/4$.
\end{proof}

Second, one can allocate a separate information budget to control the self-correlation along each covariance eigenmode.


\begin{corollary}[Spectral rate allocation]
\label{cor:spectral-allocation}
Let $\Sigma=O\diag(\sigma_1^2,\ldots,\sigma_d^2)O^\top$, where $O$ is
orthogonal and $\sigma_i>0$.  For any finite $r_1,\ldots,r_d\ge0$, there
exists $(X,Y)\in\Pi(\mu,\mu)$ such that
\begin{align}
 I(X;Y)&\le\sum_{i=1}^d r_i,\label{eq:spectral-rate}\\
 O^\top\Sym\E[(X-m)(Y-m)^\top]O
 &\succeq\diag\bigl(\sigma_i^2\rho(r_i)\bigr).
 \label{eq:spectral-cross}
\end{align}
In particular,
\begin{equation}\label{eq:spectral-trace-distortion}
 \E\norm{X-Y}^2
 \le 2\sum_{i=1}^d\sigma_i^2\bigl(1-\rho(r_i)\bigr).
\end{equation}
\end{corollary}

\begin{proof}
Choose $S=O\diag(\sigma_i^2\rho(r_i))O^\top$.  The eigenvalues of
$\tilde{S}=\Sigma^{-1/2}S\Sigma^{-1/2}$ are $\rho(r_i)$, so the right-hand side
of \eqref{eq:matrix-coupling-result} equals
$-\frac12\sum_i\log(1-\rho(r_i)^2)=\sum_i r_i$. This proves \eqref{eq:spectral-rate} and \eqref{eq:spectral-cross}. The distortion bound \eqref{eq:spectral-trace-distortion} then follows from \eqref{eq:self-cost-correlation}.
\end{proof}

\subsection{A Gaussian reference and the midpoint construction}
\label{sec:midpoint-proof}


Assume that the source has mean $m$ and positive-definite covariance
matrix $\Sigma$.  For any self-coupling $(X,Y)$ of the source law, define
the standardized pair
\[
 \widetilde X:=\Sigma^{-1/2}(X-m),
 \qquad
 \widetilde Y:=\Sigma^{-1/2}(Y-m).
\]
Both $\widetilde X$ and $\widetilde Y$ then have the same law $\nu$,
which has mean zero and covariance $I_d$.  Since the transformation
$x\mapsto \Sigma^{-1/2}(x-m)$ is bijective, it preserves mutual
information:
\[
 I(\widetilde X;\widetilde Y)=I(X;Y).
\]
Moreover, the cross-covariance transforms by congruence:
\[
 \E[\widetilde X\widetilde Y^\top]
 =
 \Sigma^{-1/2}
 \E[(X-m)(Y-m)^\top]
 \Sigma^{-1/2}.
\]
Thus, the original problem is equivalent to its standardized version with source law $\nu$ and Gaussian reference $\gamma=\mathcal N(0,I_d)$. Henceforth, we work in this normalized setting and, for simplicity, relabel $(\widetilde X,\widetilde Y)$ as $(X,Y)$.

For a symmetric matrix $S$, define the information cost
\begin{equation}\label{eq:Jmu}
 J_\nu(S):=
 \inf_{\substack{(X,Y)\in\Pi(\nu,\nu)\\
                  \Sym\E[XY^\top]\succeq S}}
 I(X;Y).
\end{equation}
We aim to show $J_\nu(S)\le J_\gamma(S)$.  
Introducing a positive-semidefinite multiplier $M$ replaces the hard correlation constraint with the reward $\E[X^\top MY]$.
 This leads to the free-energy functional
\begin{equation}\label{eq:free-energy}
 \Lambda_\nu(M):=
 \sup_{(X,Y)\in\Pi(\nu,\nu)}
 \{\E[X^\top MY]-I(X;Y)\},\qquad M\succeq0.
\end{equation}
Cone duality then gives
\begin{equation}\label{eq:fenchel}
 J_\nu(S)=\sup_{M\succeq0}
 \{\Tr(MS)-\Lambda_\nu(M)\}.
\end{equation}
The compactness and conjugacy argument is given in Appendix 
\ref{app:duality}.   It remains to prove 
$\Lambda_\nu(M)\ge\Lambda_\gamma(M)$.

Temporarily assume that $\nu=f\gamma$ and $\KL{\nu}{\gamma}<\infty$. The general case is treated in Appendix \ref{app:smoothing}. For any self-coupling $\pi$ of $\nu$,  the relative entropy chain rule gives
\begin{equation}\label{eq:MI-two-KL}
I_\pi(X;Y)
=\mathrm{KL}(\pi|\gamma\otimes\gamma)-2\mathrm{KL}(\nu|\gamma),
\end{equation}
where $\mathrm{KL}(\nu|\gamma)$ does not depend on $\pi\in\Pi(\nu,\nu)$.
Thus, up to an additive constant, the free energy is the quadratic reward minus  the relative entropy with respect to the product Gaussian law.  We absorb the quadratic reward by introducing a correlated Gaussian reference.


For $0\preceq S\prec I_d$, let $Q_S$ denote the joint law of $(G,\widetilde G)$, where
\begin{equation}\label{eq:gaussian-reference}
 G\sim\gamma,\qquad
 \widetilde G=SG+(I_d-S^2)^{1/2}E,
 \qquad E\sim\gamma\text{ independent of }G.
\end{equation}
The pair $(G,\widetilde G)$ has cross-covariance $S$.  Set
\begin{equation}\label{eq:multiplier}
 M_S=S(I_d-S^2)^{-1}.
\end{equation}
Every $M\succeq0$ can be written in this form for a unique $S$, given by
\begin{equation}\label{eq:S-from-M}
 S=2M\bigl(I_d+(I_d+4M^2)^{\frac{1}{2}}\bigr)^{-1}.
\end{equation}
On each eigenspace, the corresponding eigenvalues $m$ and $s$ satisfy $m=s/(1-s^2)$.

The inverse of the block covariance matrix
$\bigl(\begin{smallmatrix}I_d&S\\S&I_d\end{smallmatrix}\bigr)$ has
diagonal blocks $(I_d-S^2)^{-1}$ and off-diagonal blocks $-M_S$.
Consequently,
\begin{equation}\label{eq:reference-density}
 \log\frac{\dd Q_S}{\dd (\gamma\otimes\gamma)}(x,y)
 =x^\top M_Sy-\frac12x^\top SM_S x-\frac12y^\top SM_Sy
   -\frac12\log\det(I_d-S^2).
\end{equation}
Because both marginals have covariance $I_d$, the endpoint quadratic terms have a fixed expectation:
$\frac12\E[X^\top SM_SX+Y^\top SM_SY]=\Tr(SM_S)$.
Using \eqref{eq:MI-two-KL}, we obtain the exact identity
\begin{equation}\label{eq:change-reference}
 \E_\pi[X^\top M_SY]-I_\pi(X;Y)
 =2\mathrm{KL}(\nu\|\gamma)+\Tr(SM_S)
   +\frac12\log\det(I_d-S^2)-\mathrm{KL}(\pi\|Q_S).
\end{equation}
When the marginals are Gaussian, choosing $\pi=Q_S$ makes the last term zero. Hence,
\begin{equation}\label{eq:gaussian-free-energy-explicit}
 \Lambda_\gamma(M_S)=\Tr(SM_S)+\frac12\log\det(I_d-S^2).
\end{equation}
Taking the supremum in \eqref{eq:change-reference} now yields
\begin{equation}\label{eq:free-energy-exact-gap}
 \Lambda_\nu(M_S)-\Lambda_\gamma(M_S)
 =2\mathrm{KL}(\nu\|\gamma)
   -\inf_{\pi\in\Pi(\nu,\nu)}\mathrm{KL}(\pi\|Q_S).
\end{equation}
Thus, the desired comparison reduces to finding a self-coupling whose relative entropy with respect to $Q_S$ is at most the sum of the two marginal relative entropies. The midpoint construction below produces such a coupling. 


For a symmetric $0\preceq T\preceq I_d$, define the Gaussian
OU kernel
\begin{equation}\label{eq:OU-kernel-definition}
 K_T(x,\cdot)=\Law\bigl(Tx+(I_d-T^2)^{1/2}E\bigr),
 \qquad E\sim\gamma.
\end{equation}
We write $K_Tf(x)=\int f(z)K_T(x,\dd z)$ for its action on functions and
$\nu K_T$ for its action on probability laws.  The joint Gaussian law
$\gamma(\dd x)K_T(x,\dd z)$ is  symmetric in $x$ and $z$.
Thus, $K_T$ is reversible with invariant law $\gamma$.
For commuting symmetric contractions $T_1,T_2$,
\begin{equation}\label{eq:OU-semigroup}
 K_{T_1}K_{T_2}=K_{T_1T_2}.
\end{equation}
Indeed, composing the kernels combines their independent Gaussian noises, whose total covariance is
$$T_2(I_d-T_1^2)T_2+I_d-T_2^2=I_d-(T_1T_2)^2.$$
When $T=e^{-t}I_d$, the parameter $t$ is the usual OU time.

Set $T=S^{1/2}$.  The reference coupling $Q_S$ can be split into two
equal OU steps:
\[
 \widetilde Q_S(\dd x,\dd z,\dd y)
 =\gamma(\dd x)K_T(x,\dd z)K_T(z,\dd y).
\]
To change both endpoint laws to $\nu$, start with the observation law
$\nu(\dd x)K_T(x,\dd z)$.  Its midpoint marginal is
\begin{equation}\label{eq:midpoint-marginal}
 \nu_T:=\nu K_T=f_T\gamma,\qquad f_T=K_Tf.
\end{equation}
A second copy of the same observation law has the same midpoint
marginal.  We can therefore join the two copies by sampling their endpoints independently  from the same posterior conditional on the shared midpoint $z$.  Bayes' formula and
reversibility yield the posterior law of an endpoint given $z$:
\begin{equation}\label{eq:posterior-half-bridge}
 L_T^\nu(z,\dd x)=\frac{f(x)}{f_T(z)}K_T(z,\dd x).
\end{equation}
The resulting path law is
\begin{equation}\label{eq:midpoint-generative-law}
 P_S^\nu(\dd x,\dd z,\dd y)
 =\nu(\dd x)K_T(x,\dd z)L_T^\nu(z,\dd y)
 =\nu_T(\dd z)L_T^\nu(z,\dd x)L_T^\nu(z,\dd y).
\end{equation}
These two factorizations describe equivalent ways to sample the construction: draw one endpoint and then sample the other, or draw the shared midpoint and then draw two conditionally independent endpoints from the posterior. In either description, $X-Z-Y$ is a Markov chain, and both endpoints have law $\nu$. See Fig. \ref{fig:midpoint-gluing} for an illustration.

\begin{figure}[t]
\centering
\begin{tikzpicture}[
 scale=1.0,
 state/.style={draw,circle,minimum size=10mm,fill=blue!5},
 midpoint/.style={state,fill=orange!12},
 obs/.style={-{Latex[length=2mm]},thick},
 rev/.style={-{Latex[length=2mm]},thick,dashed},
 lab/.style={font=\small,align=center}]
\node[state] (wl) at (0,0) {$X$};
\node[midpoint] (zl) at (3.7,0) {$Z_1$};
\node[midpoint] (zr) at (7,0) {$Z_2$};
\node[state] (yr) at (10.7,0) {$Y$};
\draw[obs] (wl) -- node[above,lab] {OU observation $K_T$} (zl);
\draw[obs] (yr) -- node[above,lab] {OU observation $K_T$} (zr);
\node[lab,below=3mm of wl] {$\nu$};
\node[lab,below=3mm of yr] {$\nu$};
\node[lab] at (5.35,-0.75) {same midpoint law\quad $\nu_T$};
\draw[decorate,decoration={brace,mirror,amplitude=4pt}]
 (3.3,-1.15) -- (7.4,-1.15);
\node[lab] at (5.35,-1.75) {glue conditionally over $\nu_T$};
\node[state] (w) at (1.65,-3) {$X$};
\node[midpoint] (z) at (5.35,-3) {$Z$};
\node[state] (y) at (9.05,-3) {$Y$};
\draw[obs] (w) -- node[above,lab] {$K_T$} (z);
\draw[rev] (z) -- node[above,lab] {Bayes reverse $L_T^\nu$} (y);
\node[lab,below=3mm of w] {$\nu$};
\node[lab,below=3mm of z] {$\nu_T$};
\node[lab,below=3mm of y] {$\nu$};
\end{tikzpicture}
\caption{The top row shows the two observation laws, one from each endpoint; they have the same midpoint marginal. Conditional gluing yields the Markov chain shown in the bottom row. Its first transition is an OU observation, and its second is the Bayes reverse of that observation under $\nu$.}
\label{fig:midpoint-gluing}
\end{figure}
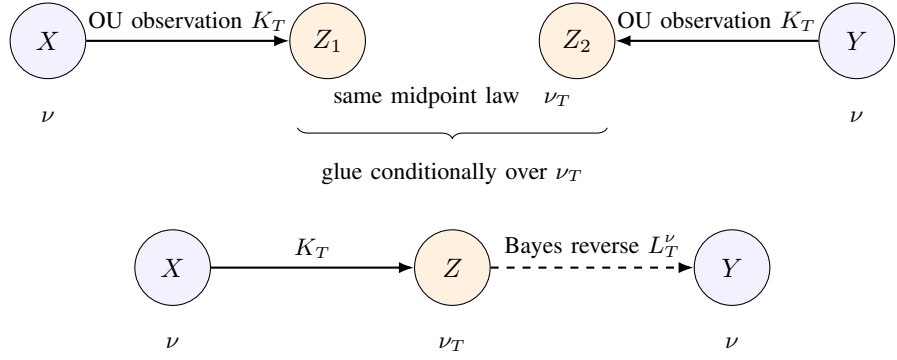

Taking the Radon–Nikodym derivative of \eqref{eq:midpoint-generative-law} with respect to the Gaussian reference path law gives
\begin{equation}\label{eq:tilted-bridge}
 \frac{\dd P_S^\nu}{\dd \widetilde Q_S}(x,z,y)
 =\frac{f(x)f(y)}{f_T(z)}.
\end{equation}
The denominator arises from the two posterior half-bridges, each of which contributes a factor $f_T(z)^{-1}$; sampling the shared midpoint from $f_T\gamma$ restores one factor of $f_T(z)$.

\subsection{The entropy bound and Gaussian comparison}
\label{subsec:midpoint-entropy}

The three marginals of the constructed path are known, so its
relative entropy can be evaluated directly from
\eqref{eq:tilted-bridge}. Projecting onto the endpoints then gives
the bound needed in \eqref{eq:free-energy-exact-gap}.
\begin{lemma}[Midpoint entropy bound]
\label{lem:ou-midpoint}
Suppose $\nu=f\gamma$ and $\mathrm{KL}(\nu\|\gamma)<\infty$.  For
$0\preceq S\prec I_d$ and $T=S^{1/2}$, the endpoint law
$\pi_S^\nu$ of \eqref{eq:midpoint-generative-law} belongs to
$\Pi(\nu,\nu)$ and satisfies
\begin{equation}\label{eq:midpoint-bound}
 \mathrm{KL}(\pi_S^\nu\|Q_S)
 \le 2\mathrm{KL}(\nu\|\gamma)-\mathrm{KL}(\nu K_T\|\gamma).
\end{equation}
\end{lemma}

\begin{proof}
The factorization in \eqref{eq:midpoint-generative-law} gives endpoint
marginals $\nu$ and midpoint marginal $\nu_T$.
The semigroup identity sends $\widetilde Q_S$ to $Q_S$ upon projection
onto $(X,Y)$.  Thus, the data processing inequality and the three marginal identities
yield
\begin{align*}
 \mathrm{KL}(\pi_S^\nu\|Q_S)
 &\le \mathrm{KL}(P_S^\nu\|\widetilde Q_S)\\
 &=\E_{P_S^\nu}[\log f(X)+\log f(Y)-\log f_T(Z)]\\
 &=2\mathrm{KL}(\nu\|\gamma)-\mathrm{KL}(\nu_T\|\gamma).
\end{align*}
All terms are finite because relative entropy contracts under the
kernel $K_T$.
\end{proof}

The relative entropy of each endpoint is counted once, while that of the shared midpoint is subtracted once. This midpoint term provides precisely the gain needed in \eqref{eq:free-energy-exact-gap}.

\begin{proposition}[Free-energy comparison]
\label{prop:free-energy}
Let $\nu$ have mean zero and covariance $I_d$.  For every $M\succeq0$,
let $S$ be given by \eqref{eq:S-from-M}.  Then
\begin{equation}\label{eq:free-energy-comparison}
 \Lambda_\nu(M)
 \ge\Lambda_\gamma(M)+\mathrm{KL}(\nu K_{S^{1/2}}\|\gamma)
 \ge\Lambda_\gamma(M).
\end{equation}
\end{proposition}

\begin{proof}
When $\mathrm{KL}(\nu\|\gamma)<\infty$, substituting
\eqref{eq:midpoint-bound} into \eqref{eq:free-energy-exact-gap} yields the desired comparison.
For a general law, \Cref{app:smoothing} extends 
 the argument using covariance-preserving Gaussian smoothing and the joint lower semicontinuity of mutual information and midpoint relative entropy.
\end{proof}

\begin{proof}[Proof of \Cref{thm:matrix-self-coupling}]
    Returning to the correlation constraint reverses the direction of the comparison.   Equations \eqref{eq:fenchel} and
\eqref{eq:free-energy-comparison} yield
\begin{equation}\label{eq:J-order}
 J_\nu(S)\le J_\gamma(S).
\end{equation}
To evaluate the Gaussian side, \eqref{eq:reference-density} gives, for
every coupling $\pi\in\Pi(\gamma,\gamma)$,
\begin{equation}\label{eq:gaussian-KL-identity}
 I_\pi(X;Y)
 =-\frac12\log\det(I_d-S^2)
 +\Tr\!\left(M_S\bigl(\Sym\E_\pi[XY^\top]-S\bigr)\right)
 +\mathrm{KL}(\pi\|Q_S).
\end{equation}
If $\Sym\E_\pi[XY^\top]\succeq S$, the trace and relative-entropy terms are nonnegative.  The coupling $Q_S$ makes both terms zero.  Thus,
\begin{equation}\label{eq:gaussian-cost}
 J_\gamma(S)=-\frac12\log\det(I_d-S^2),
 \qquad 0\preceq S\prec I_d.
\end{equation}
Together with the attainment result  in Appendix \ref{app:duality}, this proves the
matrix comparison in \Cref{thm:matrix-self-coupling} after undoing the
affine normalization.

\textbf{Strictness from the midpoint relative entropy gain.}
Fix $0\preceq S\prec I_d$ with $\KL{\nu K_{S^{1/2}}}{\gamma}>0$.
We prove  $J_\nu (S) < J_\gamma (S)$ by contradiction, using a sequence of  approximately optimal multipliers in the dual representation.
Suppose instead that $J_\nu(S)=J_\gamma(S)$. For each $\epsilon>0$, choose
an $\epsilon$-optimal multiplier in \eqref{eq:fenchel}:
\[
 \Tr(M_\epsilon S)-\Lambda_\nu(M_\epsilon)
 \ge J_\nu(S)-\epsilon.
\]
Let $B_\epsilon$ denote the corresponding Gaussian reference correlation, so that
$M_\epsilon=B_\epsilon(I_d-B_\epsilon^2)^{-1}$.
Applying \eqref{eq:gaussian-KL-identity} with coupling $\pi=Q_S$ and reference
$Q_{B_\epsilon}$, and using
\eqref{eq:gaussian-free-energy-explicit}, gives the exact dual gap
\begin{equation}\label{eq:gaussian-dual-gap}
 J_\gamma(S)-\Tr(M_\epsilon S)+\Lambda_\gamma(M_\epsilon)
 =\KL{Q_S}{Q_{B_\epsilon}},
\end{equation}
Together with \eqref{eq:free-energy-comparison}, this yields
\begin{equation}\label{eq:strictness-two-entropies}
 \begin{aligned}
 \epsilon
 &\ge J_\gamma(S)-\Tr(M_\epsilon S)+\Lambda_\nu(M_\epsilon)\\
 &=\KL{Q_S}{Q_{B_\epsilon}}
   +\Lambda_\nu(M_\epsilon)-\Lambda_\gamma(M_\epsilon)\\
 &\ge\KL{Q_S}{Q_{B_\epsilon}}
   +\KL{\nu K_{B_\epsilon^{1/2}}}{\gamma}.
 \end{aligned}
\end{equation}
Thus, both nonnegative divergences tend to zero as $\epsilon\downarrow0$.  By Pinsker's inequality, $Q_{B_\epsilon}$ converges to $Q_S$ in total variation, and hence weakly. 
 Since these Gaussian laws have fixed marginals, their cross-moments are uniformly integrable, as shown in \Cref{app:duality}; it follows that  $B_\epsilon\to S$.
Continuity of the OU observation then gives
$\nu K_{B_\epsilon^{1/2}}\rightharpoonup\nu K_{S^{1/2}}$.
By lower semicontinuity of relative entropy, 
\[
 0<\KL{\nu K_{S^{1/2}}}{\gamma}
 \le\liminf_{\epsilon\downarrow0}
       \KL{\nu K_{B_\epsilon^{1/2}}}{\gamma}=0,
\]
a contradiction. Hence $J_\nu(S)<J_\gamma(S)$ whenever the midpoint relative entropy is positive.

For non-Gaussian $\nu$ and $S\succ0$, the midpoint relative entropy $\KL{\nu K_{S^{1/2}}}{\gamma}$ is strictly positive. Indeed, if $\nu K_{S^{1/2}}=\gamma$, then equality of characteristic functions gives
\[
 \widehat\nu(S^{1/2}t)e^{-t^\top(I_d-S)t/2}
 =e^{-\norm{t}^2/2},
 \qquad
 \widehat\nu(S^{1/2}t)=e^{-t^\top St/2}.
\]
Since $S^{1/2}$ is invertible, it follows that $\nu=\gamma$, a contradiction. 
This proves the strict clause of \Cref{thm:matrix-self-coupling}.
\end{proof}

\begin{proof}[Proof of \Cref{cor:scalar-self-correlation}]
Center and rescale the source to obtain a law $\nu$ with variance one. At $R=0$, independent copies have covariance zero. For $R>0$, let $s=\rho(R)\in(0,1)$. Theorem \ref{thm:matrix-self-coupling} in dimension one gives $J_\nu(s)\le J_\gamma(s)=R$, and attainment of the minimum in \eqref{eq:matrix-coupling-result}  provides the required self-coupling. For the Gaussian law, the maximum covariance is exactly $s$ by \eqref{eq:gaussian-cost}.

If $\nu$ is non-Gaussian, the strict part of Theorem \ref{thm:matrix-self-coupling} gives $J_\nu(s)<R$. Since $J_\nu$ is convex and finite on $(0,1)$, it is continuous there. We can therefore choose $s'\in(s,1)$ sufficiently close to $s$ that $J_\nu(s')<R$. A minimizing coupling for $J_\nu(s')$ has mutual information below $R$ and covariance at least $s'>\rho(R)$. Undoing the centering and rescaling multiplies the covariance by $\sigma^2$, proving \eqref{eq:strict-scalar-self-correlation}.
\end{proof}



\section{Gaussian extremality for rate--distortion--perception}
\label{sec:rdp}

This section establishes the RDP counterpart of the well-known Gaussian extremality result in RD theory, using the self-coupling inequality. Throughout this section, we assume that the source covariance is positive definite.


\subsection{Problem formulation and Gaussian comparison}

For $\mu,\nu\in\mcP_2(\R^d)$, the squared Wasserstein-2 distance is
\begin{equation}\label{eq:W2-definition}
 W_2^2(\mu,\nu)=\inf_{\pi\in\Pi(\mu,\nu)}\E_\pi\norm{X-Y}^2.
\end{equation}
For  $\mu\in\mcP_2(\R^d)$, define
\begin{equation}\label{eq:rdp-unlimited}
 D_\mu(R,P)
 :=\inf_{\substack{X\sim\mu,\ I(X;\widehat X)\le R\\
 W_2^2(\mu,\Law(\widehat X))\le P}}
       \E\norm{X-\widehat X}^2,
 \qquad R,P\ge0.
\end{equation}
The infimum is taken over joint laws of $(X,\widehat X)$ subject to the specified constraints.  This is the quadratic Wasserstein version of the
RDP function introduced by Blau and Michaeli
\cite{BlauMichaeli2019}. At $P=0$, the reconstruction must have
law $\mu$, so \eqref{eq:rdp-unlimited} reduces to the information-constrained
self-coupling problem.

We recall the Gaussian benchmark established by Qian et al. \cite[Theorem~5]{Qian25} and Qu et al.
\cite[Theorem~3]{QuChenYuXu2025}. Let $G=\mcN(m,\Sigma)$ with
$\Sigma\succ0$, and let
$\sigma_1^2,\ldots,\sigma_d^2$ be the eigenvalues of $\Sigma$.  Then
\begin{equation}\label{eq:gaussian-rdp}
\begin{split}
 D_G(R,P)
 =\min_{\substack{r_i\ge0\\ \sum_i r_i\le R}}
 \biggl\{&\sum_{i=1}^d\sigma_i^2e^{-2r_i}+\pos{\sqrt{\sum_{i=1}^d\sigma_i^2(1-\rho(r_i))^2}
             -\sqrt P}^{\,2}\biggr\}.
\end{split}
\end{equation}
In particular \cite[Theorem 1]{Zhang25}, for a scalar Gaussian source with  variance $\sigma^2>0$,
\begin{equation}\label{eq:scalar-rdp}
 D_G(R,P)
 =\sigma^2e^{-2R}
   +\pos{\sigma(1-\rho(R))-\sqrt P}^{\,2}.
\end{equation}
We use this formula  to compare an arbitrary source with the Gaussian source having the same mean and covariance.

It is instructive to begin in one dimension.  Suppose  $\mu$ has mean
$m$ and variance $\sigma^2>0$, and recall that $\rho(r)=\sqrt{1-e^{-2r}}$.
By \Cref{cor:scalar-self-correlation}, there is a self-coupling $(X,Y)$ such that
\[
 I(X;Y)\le R,
 \qquad
 \E[(X-m)(Y-m)]\ge\sigma^2\rho(R).
\]
For $0\le\tau\le1$, define the reconstruction
$\widehat X=m+\tau(Y-m)$.  By the data processing inequality, its information cost is at most $R$. Coupling $Y$ with its scaled version gives
\begin{align}
 W_2^2(\mu,\Law(\widehat X))
 &\le\sigma^2(1-\tau)^2,
 \label{eq:rdp-scalar-perception}\\
 \E[(X-\widehat X)^2]
 &\le\sigma^2\bigl(1+\tau^2-2\tau\rho(R)\bigr)
 =\sigma^2e^{-2R}+\sigma^2\bigl(\tau-\rho(R)\bigr)^2.
 \label{eq:rdp-scalar-scaling}
\end{align}
The distortion bound is minimized at scale $\rho(R)$, while the perception constraint requires $\tau$ to remain sufficiently close to $1$. The optimal choice is  
\begin{equation}\label{eq:rdp-scalar-scale}
 \tau=\max\left\{\rho(R),\,1-\frac{\sqrt P}{\sigma}\right\}.
\end{equation}
Substituting this choice gives the Gaussian benchmark in \eqref{eq:scalar-rdp}.
In higher dimensions, the same construction scales each covariance eigendirection, while the perception budget controls the total squared displacement.

\begin{theorem}[Gaussian extremality for RDP]
\label{thm:rdp-extremality}
Let $\mu\in\mcP_2(\R^d)$ have mean $m$ and covariance $\Sigma\succ0$, and let
$G=\mcN(m,\Sigma)$.  For every $R,P\ge0$,
\begin{equation}\label{eq:rdp-comparison}
 D_\mu(R,P)\le D_G(R,P).
\end{equation}
If $\mu$ is non-Gaussian and the minimization in
\eqref{eq:gaussian-rdp} admits an optimal allocation with $r_i>0$ for every $i$, then the inequality in 
\eqref{eq:rdp-comparison} is strict.  

\end{theorem}


\subsection{Proof by endpoint scaling}

\begin{proof}[Proof of \Cref{thm:rdp-extremality}]
By translation and orthogonal invariance, we may assume  without loss of generality that $m=0$ and
$\Sigma=\diag(\sigma_1^2,\ldots,\sigma_d^2)$.

To construct a reconstruction, fix an allocation $r_i\ge0$ with $\sum_i r_i\le R$, and set
$\rho_i=\rho(r_i)$.  By \Cref{thm:matrix-self-coupling}, there is a
self-coupling $(X,Y)$ satisfying
\begin{equation}\label{eq:rdp-coordinate-coupling}
 I(X;Y)\le\sum_i r_i,
 \qquad
 \E[X_iY_i]\ge\sigma_i^2\rho_i
 \quad(1\le i\le d).
\end{equation}
For nonnegative scales $\tau_i$, define $\widehat X_i=\tau_iY_i$.  Then
\begin{align}
 I(X;\widehat X)&\le R,\nonumber\\
 W_2^2(\mu,\Law(\widehat X))
 &\le\sum_i\sigma_i^2(1-\tau_i)^2,
 \label{eq:rdp-vector-perception}\\
 \E\norm{X-\widehat X}^2
 &\le\sum_i\sigma_i^2e^{-2r_i}
       +\sum_i\sigma_i^2(\tau_i-\rho_i)^2.
 \label{eq:rdp-vector-scaling}
\end{align}
Choosing the scales amounts to a Euclidean projection problem. Define
\[
 v=(\sigma_i)_{i=1}^d,
 \qquad z=(\sigma_i\rho_i)_{i=1}^d,
 \qquad W_0=\norm{v-z}.
\]
We seek the point $x=(\sigma_i\tau_i){i=1}^d$ in the closed ball of radius $\sqrt P$ centered at $v$ that is closest to $z$. If $W_0\le\sqrt P$, choose $x=z$. Otherwise, choose the point on the segment from $v$ to $z$ at distance $\sqrt P$ from $v$. Equivalently,
\begin{equation}\label{eq:rdp-vector-scales}
 \tau_i=
 \begin{cases}
  \rho_i,&W_0\le\sqrt P,\\[1mm]
  1-\dfrac{\sqrt P}{W_0}(1-\rho_i),&W_0>\sqrt P.
 \end{cases}
\end{equation}
These scales lie in $[\rho_i,1]$, satisfy the perception constraint, and yield
\begin{equation}\label{eq:rdp-vector-value}
 \E\norm{X-\widehat X}^2
 \le\sum_i\sigma_i^2e^{-2r_i}
      +\pos{W_0-\sqrt P}^{\,2}.
\end{equation}
Minimizing over the rate allocation bounds $D_\mu(R,P)$ by the right-hand
side of \eqref{eq:gaussian-rdp}.

To establish strictness, suppose that $\mu$ is non-Gaussian and choose a minimizing allocation with $r_i>0$ for every $i$.   Let
$\nu=\Law(\Sigma^{-1/2}X)$ and set $S=\diag(\rho_i)$, so that
$0\prec S\prec I_d$.  The strict part of
\Cref{thm:matrix-self-coupling} gives $J_\nu(S)<\sum_i r_i\le R$.
Since $J_\nu$ is convex and finite, it is continuous on $0\prec S\prec I_d$. 
Hence, for some $\delta>0$, we have $S+\delta I_d\prec I_d$ and $J_\nu(S+\delta I_d)<R$. By attainment, there is a self-coupling satisfying
$\E[X_iY_i]\ge\sigma_i^2(\rho_i+\delta)$ for every $i$.
Using this coupling with the same positive scales as in  \eqref{eq:rdp-vector-scales} gives
\[
 \E\norm{X-\widehat X}^2
 \le D_G(R,P)-2\delta\sum_i\sigma_i^2\tau_i<D_G(R,P).
\]
\end{proof}

\begin{remark}[Equality for a non-Gaussian vector source]
\label{rem:rdp-nonstrict-vector}
Non-Gaussian vector sources can attain equality in \eqref{eq:rdp-comparison} when the Gaussian optimum assigns zero rate to some directions. Appendix \ref{app:rdp-nonstrict-vector} gives a two-dimensional example with a complete calculation.
\end{remark}

We next present an alternative construction based on posterior-mean interpolation. We consider two candidate  representations: the conditional mean induced by a self-coupling and a conventional RD optimizer. The first yields another proof of \eqref{eq:rdp-comparison}; the second is discussed in \Cref{rem:conventional-rd-interpolation}.

\subsection{Posterior-mean interpolation}
\label{subsec:posterior-interpolation}

The following construction builds on the Wasserstein geometry of the distortion–perception tradeoff developed in \cite{Freirich2021}.


\subsubsection{The interpolation curve}

Suppose that $X\sim\mu$ and $U$ satisfy
$\E[X\mid U]=U$ and $I(X;U)\le R$. Define
\[
 \varepsilon=\E\norm{X-U}^2,\qquad
 w=W_2(\mu,\Law(U)).
\]
Extend $(X,U)$ to a Markov chain $X-U-Y$ by sampling $Y$ given $U$ according to the conditional law induced by an optimal Wasserstein-2 coupling of $\Law(U)$ and $\mu$. Then $Y\sim\mu$ and $\E\norm{Y-U}^2=w^2$.
For  $P\ge0$, define
\begin{equation}\label{eq:appendix-posterior-decoder}
 \widehat X=
 \begin{cases}
  U,&w\le\sqrt P,\\[1mm]
  \dfrac{\sqrt P}{w}U+
       \left(1-\dfrac{\sqrt P}{w}\right)Y,&w>\sqrt P.
 \end{cases}
\end{equation}
The first case includes $w=0$. By the Markov property and the data processing inequality,
$I(X;\widehat X)\le I(X;U,Y)=I(X;U)\le R$.
When $w>\sqrt P$, coupling $\widehat X$ with $Y$ gives
\[
 W_2^2(\mu,\Law(\widehat X))
 \le\E\norm{Y-\widehat X}^2=P.
\]
When $w\le\sqrt P$, the perception constraint follows from $w^2\le P$.
Moreover, conditional independence and the conditional-mean property imply that $\E[X-U\mid U,Y]=0$. Hence the cross term in the distortion vanishes, and
\begin{equation}\label{eq:posterior-interpolation-curve}
 \E\norm{X-\widehat X}^2
 =\varepsilon+\E\norm{U-\widehat X}^2
 =\varepsilon+\pos{w-\sqrt P}^{\,2}.
\end{equation}
This identity illustrates the geometry: $\varepsilon$ is the estimation error incurred at $U$, while the second term is the cost of moving its law toward the source law. Under perfect perception, the distortion is $\varepsilon+w^2$; when $P\ge w^2$, it is $\varepsilon$.

\subsubsection{Endpoint bounds and Gaussian comparison}

To compare this curve with the Gaussian benchmark, we need bounds on both the estimation error $\varepsilon$ and the distortion under perfect perception, $\varepsilon+w^2$. The self-coupling provides both bounds simultaneously.

Work first in centered coordinates, with
$\Sigma=\diag(\sigma_1^2,\ldots,\sigma_d^2)\succ0$.
Fix $r_i\ge0$ such that $\sum_i r_i\le R$, and set $\rho_i=\rho(r_i)$.
Choose $(X,Y_0)$ as in \eqref{eq:rdp-coordinate-coupling}, and define
$U=\E[X\mid Y_0]$. We have $\E[X\mid U]=U$ by the tower property of conditional expectation
and  $I(X;U)\le R$ by the data processing inequality.
For each coordinate,
\[
 \E[U_iY_{0,i}]=\E[X_iY_{0,i}]\ge\sigma_i^2\rho_i.
\]
Cauchy--Schwarz therefore gives
$\E[U_i^2]\ge\sigma_i^2\rho_i^2$.
Define the two quantities associated with the Gaussian allocation:
\[
 D_0=\sum_i\sigma_i^2(1-\rho_i^2),\qquad
 W_0^2=\sum_i\sigma_i^2(1-\rho_i)^2.
\]
Conditional-mean orthogonality yields the first bound:
\begin{equation}\label{eq:interpolation-first-endpoint}
 \varepsilon=\sum_i\bigl(\sigma_i^2-\E[U_i^2]\bigr)\le D_0.
\end{equation}
For the second, $(U,Y_0)$ couples $\Law(U)$ and $\mu$,  so $w^2\le\E\norm{Y_0-U}^2$.
Since $\E[X-U\mid Y_0]=0$ and $Y_0-U$ is a function of $Y_0$,
we have $\E\langle X-U,Y_0-U\rangle=0$. Expanding the square gives
\begin{equation}\label{eq:interpolation-second-endpoint}
 \begin{split}
 \varepsilon+w^2
 &\le \E\norm{X-U}^2+\E\norm{Y_0-U}^2\\
 &=\E\norm{X-Y_0}^2\\
 &\le 2\sum_i\sigma_i^2(1-\rho_i)=D_0+W_0^2.
 \end{split}
\end{equation}
We use $(U,Y_0)$ to bound the Wasserstein distance, while the reconstruction in \eqref{eq:appendix-posterior-decoder} uses the optimal coupling $(U,Y)$ introduced above.

To interpolate between these endpoint estimates for an arbitrary perception tolerance, define
 $h_P(v)=(\sqrt v-\sqrt P)_+^2$ for $v,P\ge0$. The identity
\begin{equation}\label{eq:perception-variational}
 h_P(v)=\sup_{0\le\theta<1}
 \left\{\theta v-\frac{\theta P}{1-\theta}\right\}
\end{equation}
expresses this penalty as a supremum of weighted squared-distance terms.
For $P>0$, the supremum is attained at $\theta=0$ if $v\le P$, and at $\theta=1-\sqrt{P/v}$ if $v>P$. When $P=0$, it approaches $v$ as $\theta\uparrow1$.
For each $0\le\theta<1$, the two endpoint estimates imply
\[
 \varepsilon+\theta w^2
 =(1-\theta)\varepsilon+\theta(\varepsilon+w^2)
 \le D_0+\theta W_0^2.
\]
Subtracting $\theta P/(1-\theta)$ and taking the supremum over $\theta$ on both sides, the identity 
 \eqref{eq:perception-variational} yields
\begin{equation}\label{eq:interpolation-gaussian-allocation}
 \varepsilon+\pos{w-\sqrt P}^{\,2}
 \le D_0+\pos{W_0-\sqrt P}^{\,2}.
\end{equation}
The case $P=0$ follows by taking $\theta\uparrow1$. Minimizing the right-hand side over rate allocations and applying \eqref{eq:gaussian-rdp} proves
$D_\mu(R,P)\le D_G(R,P)$. Undoing the translation and orthogonal change of coordinates gives the same conclusion for an arbitrary mean and positive-definite covariance.

\begin{remark}[Starting from a conventional RD optimizer]
\label{rem:conventional-rd-interpolation}
One can also start with an optimal reconstruction $\widehat{X}_{\rm RD}$ for the ordinary quadratic RD problem: 
 $$\min_{P_{\widehat{X}|X}:\,I(X;\widehat{X})\le R}
    \E\norm{X-\widehat{X}}^2.$$
 Note that   $\widehat{X}_{\rm RD}= \E[X| \widehat{X}_{\rm RD}]$. 
Writing $\varepsilon_{\rm RD}=\E\norm{X-\widehat{X}_{\rm RD}}^2$ and
$w_{\rm RD}=W_2(\mu,\Law(\widehat{X}_{\rm RD}))$, the same interpolation yields
$\varepsilon_{\rm RD}+\pos{w_{\rm RD}-\sqrt P}^{\,2}$.
For a scalar source of variance $\sigma^2$, the classical Gaussian
RD bound gives
$\varepsilon_{\rm RD}\le\sigma^2e^{-2R}$.
To establish the full RDP comparison using the preceding endpoint argument,
it would suffice to show that
\begin{equation}\label{eq:rd-optimizer-endpoint}
 \varepsilon_{\rm RD}+w_{\rm RD}^2
 \le 2\sigma^2(1-\rho(R)).
\end{equation}
This requires additional control of the reconstruction law at $P=0$; the ordinary RD bound controls only the first endpoint. In multi-dimensional cases, a fixed ordinary RD optimizer can be suboptimal for the RDP problem under perfect perception, even for a Gaussian source, because the optimal rate allocation depends on the perception tolerance. \Cref{app:fixed-rd-representation} gives an explicit two-dimensional example.
\end{remark}

\section{Gaussian extremality with limited common randomness}
\label{sec:rdpc-formulation}

Unlike the classical RD problem, the fundamental RDP limit depends on the amount of common randomness available. As shown in \cite{Chen22}, under the sequence-level perception constraint, the RDP function introduced by Blau and Michaeli \cite{BlauMichaeli2019}  corresponds to the unlimited-common-randomness setting. We now extend the Gaussian extremality result to the RDP problem with limited common randomness, which we refer to as the RDPC problem. Throughout this section, the source law $\mu$ has mean $m$ and a positive-definite covariance matrix $\Sigma\succ0$.



\subsection{Problem formulation and Gaussian comparison targets}
\label{subsec:rdpc-formulation}

\subsubsection{Operational formulation and common auxiliary}

Consider an i.i.d. source $X^n$ with marginal law $\mu$. A distortion target $D$ is achievable under communication rate $R$, common-randomness rate $C$, and normalized squared Wasserstein-2 perception constraint $P$ if, for some blocklength $n$, there is a code with at most $\exp(nR)$ messages and a shared random seed taking at most $\exp(nC)$ values such that the reconstruction $\widehat X^n$ satisfies 
\[
\frac1n\E\norm{X^n-\widehat X^n}^2\le D,
\qquad
\frac1n W_2^2\bigl(\mu^{\otimes n},\Law(\widehat X^n)\bigr)\le P.
\]
Let $D_\mu(R,C,P)$ denote the infimum of $D$ over all such codes and blocklengths.


As shown by Qu et al. \cite[Theorem~1]{QuChenYuXu2025}, $D_\mu(R,C,P)$ has the following single-letter characterization:
\begin{equation}\label{eq:rdpc-common-representation}
 \begin{split}
 D_\mu(R,C,P)=
 \inf\biggl\{&\E\norm{X-M}^2+\pos{\sqrt{\E\norm{Y-M}^2}-\sqrt P}^{\,2}\biggr\},
 \end{split}
\end{equation}
where $(X,M)$ and $(Y,M)$ have the same marginal law for $M$ and satisfy
\begin{equation}\label{eq:rdpc-certificate-constraints}
 X,Y\sim\mu,\qquad \E[X\mid M]=M,\qquad
 I(X;M)\le R,\quad I(Y;M)\le R+C.
\end{equation}
The variable $M$ is a conditional-mean representation of the source. The endpoint $Y$ has the source law and serves as an anchor for perception. The two squared-error terms in \eqref{eq:rdpc-common-representation} measure the source-to-representation error and the anchor's displacement from that representation, respectively. In the unlimited-common-randomness limit, $D_\mu(R,C,P)$ reduces to the function $D_\mu(R,P)$ studied in \Cref{sec:rdp}:
\begin{equation}\label{eq:rdp-rdpc-identification}
 D_\mu(R,\infty,P)=D_\mu(R,P).
\end{equation}

For comparison, we recall the Gaussian formula from
\cite{QuChenYuXu2025}.  It can also be derived by combining the Gaussian weighted representation in \Cref{app:gaussian-structure} with the variational representation in \Cref{app:rdpc-variational}.
Let $G=\mcN(m,\Sigma)$, and let $\sigma_1^2,\ldots,\sigma_d^2$ be the
eigenvalues of $\Sigma\succ0$.
With $\rho(r)=\sqrt{1-e^{-2r}}$, for finite $R,C\ge0$ and $P\ge0$, we have
\begin{equation}\label{eq:gaussian-parallel}
\begin{split}
 D_G(R,C,P)
 =\min_{\substack{r_i,s_i\ge0\\
           \sum_i r_i\le R,\ \sum_i s_i\le R+C}}
 \Biggl\{&\sum_i\sigma_i^2e^{-2r_i}+\pos{\sqrt{\sum_i\sigma_i^2
       \bigl(1+\rho(r_i)^2-2\rho(r_i)\rho(s_i)\bigr)}
                -\sqrt P}^{\,2}\Biggr\}.
\end{split}
\end{equation}
The rates are allocated across covariance eigendirections: $r_i$ controls the posterior estimation error, while $s_i$ controls how closely the second endpoint aligns with the same posterior mean. This spectral formula holds for every positive-definite $\Sigma$. In dimension one, the minimum is attained by using the full rate budgets, $r=R$ and $s=R+C$, yielding
\begin{equation}\label{eq:gaussian-scalar-rdpc}
 D_G(R,C,P)
 =\sigma^2e^{-2R}
  +\pos{\sigma\sqrt{1+\rho(R)^2-2\rho(R)\rho(R+C)}
                     -\sqrt P}^{\,2}.
\end{equation}
The cases $C=0$, $P=0$, and $C=\infty$ are discussed in
\Cref{app:gaussian-special-cases}.

One possible approach to comparing a general source with the Gaussian benchmark is to start with a pair $(X,M)$ satisfying $X\sim\mu$, $\E[X\mid M]=M$, and $I(X;M)\le R$. One could then couple $\Law(M)$ with $\mu$ to obtain a pair $(Y,M)$ satisfying $I(Y;M)\le R+C$. A natural hope is that the correlation $\E[YM]$ can be controlled using only the marginal variances.
For a centered scalar source with variance $\sigma^2$, the construction in \Cref{subsec:posterior-interpolation} gives such a pair $(X,M)$ with $q:=\E M^2\ge\sigma^2\rho(R)^2$. If the bound \eqref{eq:false-two-marginal-comparison} proposed below  held, applying it to $\mu$ and $\Law(M)$ at rate $R+C$ would yield a coupling $(Y,M)$ with $I(Y;M)\le R+C$ and
$\E[YM]\ge\sigma\sqrt q \,\rho(R+C)$. Substituting these into \eqref{eq:rdpc-common-representation} would give
\[
 D_\mu(R,C,P)\le F(q)
 :=\sigma^2-q+
 \pos{\sqrt{\sigma^2+q-2\sigma\sqrt q\,\rho(R+C)}-\sqrt P}^{\,2}.
\]
By \eqref{eq:perception-variational}, $F$ is nonincreasing on
$[0,\sigma^2]$.  Thus,
$D_\mu(R,C,P)\le F(\sigma^2\rho(R)^2)=D_G(R,C,P)$,
with the Gaussian value given in \eqref{eq:gaussian-scalar-rdpc}.
The following example shows that this variance-based shortcut fails, even for scalar marginals.

\begin{remark}[A two-marginal obstruction]
\label{rem:two-marginal-false}
For arbitrary centered marginals with variances $\sigma_A^2$ and $\sigma_B^2$, consider the following extension of the self-coupling inequality:
\begin{equation}\label{eq:false-two-marginal-comparison}
 \sup_{\substack{\lambda\in\Pi(P_A,P_B)\\ I_\lambda(A;B)\le r}}
       \E_\lambda[AB]
 \stackrel{?}{\ge}\sigma_A\sigma_B\sqrt{1-e^{-2r}}.
\end{equation}
Let $A$ be uniform on $\{-1,1\}$ and let $B$ be standard Gaussian. Since $AB\le |B|$ pointwise for every coupling,
$\E[AB]\le\E|B|=\sqrt{2/\pi}$. The coupling $A=\operatorname{sgn}(B)$ attains this bound and has $I(A;B)=H(A)=\log 2$. Consequently,
\begin{equation}\label{eq:rad-gauss-max-cov}
 \sup_{\substack{\lambda\in\Pi(P_A,P_B)\\ I_\lambda(A;B)\le\log 2}}
       \E_\lambda[AB]
 =\sqrt{\frac{2}{\pi}}
 <\frac{\sqrt3}{2}
 =\sqrt{1-e^{-2\log 2}}.
\end{equation}
Thus, \eqref{eq:false-two-marginal-comparison} fails even when both marginals have variance one. This motivates constructing the common representation and its two endpoint couplings jointly, while preserving the conditional-mean structure in \eqref{eq:rdpc-certificate-constraints}.
\end{remark}

Conditional gluing provides a useful way to construct the two endpoint couplings. Let $\cA_\mu(R,C)$ denote the class of chains
\begin{equation}\label{eq:rdpc-admissible-chain}
 X-U-Y,\qquad X,Y\sim\mu,\qquad
 I(X;U)\le R,\quad I(Y;U)\le R+C,
\end{equation}
where $U$ takes values in an arbitrary standard Borel space. For a chain in this class, define
\begin{equation}\label{eq:rdpc-two-errors}
 M_U=\E[X\mid U],\qquad
 d_0=\E\norm{X-M_U}^2,\qquad
 d_1=\E\norm{Y-M_U}^2.
\end{equation}
By the tower property of conditional expectation, $\E[X\mid M_U]=M_U$, and by the data processing inequality, the information constraints in  \eqref{eq:rdpc-certificate-constraints}   hold with $M=M_U$.
Conversely, take any two feasible joint laws in \eqref{eq:rdpc-certificate-constraints} and disintegrate them as $P_{XM}=P_MK_X$ and $P_{YM}=P_MK_Y$. The joint law
$$P_M(\dd m) K_X(m, \dd x) K_Y(m, \dd y)$$
preserves both pair laws and makes $X$ and $Y$ conditionally independent given $M$. It therefore defines a chain in \eqref{eq:rdpc-admissible-chain}  with $U=M$. Consequently,
\begin{equation}\label{eq:rdpc-implicit-chain}
 D_\mu(R,C,P)=\inf_{\cA_\mu(R,C)}
 \bigl\{d_0+h_P(d_1)\bigr\},
 \qquad h_P(w)=(\sqrt w-\sqrt P)_+^2.
\end{equation}
This formulation lets us design $U$ first and recover the conditional-mean representation afterward.

The main comparison for this problem is as follows.

\begin{theorem}[Gaussian extremality for RDPC]
\label{thm:rdpc-extremality}
Let $\mu\in\mcP_2(\R^d)$ have mean $m$ and covariance $\Sigma\succ0$, and
let $G=\mcN(m,\Sigma)$. Then
\begin{equation}\label{eq:rdpc-extremality}
 D_\mu(R,C,P)\le D_G(R,C,P)
\end{equation}
for all finite $R,P\ge0$ and all $C\in[0,\infty]$.
If $\mu$ is non-Gaussian and $C<\infty$, the inequality is strict whenever the minimization in \eqref{eq:gaussian-parallel} admits a minimizing allocation with $r_i>0$ in every covariance eigendirection:
\begin{equation}\label{eq:rdpc-strict}
 D_\mu(R,C,P)<D_G(R,C,P).
\end{equation}
For $C=\infty$, the same strict comparison holds whenever a minimizing allocation in \eqref{eq:gaussian-rdp} has $r_i>0$ in every covariance eigendirection. In particular, for a non-Gaussian scalar source, the comparison is strict for every $R>0$, $P\ge0$, and $C\in[0,\infty]$.
\end{theorem}

The proof of Theorem \ref{thm:rdpc-extremality}  rests on comparing two weighted quadratic errors.

\subsubsection{A weighted formulation of the perception constraint}

The objective in \eqref{eq:rdpc-implicit-chain} depends nonlinearly on the second error $d_1$ through $h_P(d_1)$. The variational identity \eqref{eq:perception-variational}, established
in \Cref{subsec:posterior-interpolation}, expresses this penalty as a supremum of affine functions of $d_1$. This suggests comparing the weighted errors $d_0 + \theta d_1$ for each fixed $\theta$.


Define the weighted error
\begin{equation}\label{eq:weighted-energy}
 \cE_{\mu,\theta}(R,C)
 :=\inf_{\cA_\mu(R,C)}
 \left\{\E\norm{X-M_U}^2
       +\theta\E\norm{Y-M_U}^2\right\},
 \quad 0\le\theta\le1.
\end{equation}
For fixed $R,C$, let $\mathcal D_{\rm att}\subset\R^2$ be the set of error pairs $(d_0,d_1)$ generated by chains in $\cA_\mu(R,C)$, and let $\mathcal D$ be its Euclidean closure. Since both $d_0+h_P(d_1)$ and $d_0+\theta d_1$ are continuous in $(d_0,d_1)$, their infima over $\mathcal D_{\rm att}$ agree with their infima over $\mathcal D$. Time sharing makes the error region convex, while uniform bounds on the two errors make $\mathcal D$ compact. Applying the minimax theorem to \eqref{eq:perception-variational} over $\mathcal D$ gives the following identity. The details, including the endpoint case $P=0$, are provided in \Cref{app:rdpc-variational}.

\begin{proposition}[Weighted variational representation of RDPC]
\label{prop:rdpc-variational}
For every finite $R,C\ge0$ and $P>0$,
\begin{equation}\label{eq:rdpc-energy-duality}
 D_\mu(R,C,P)
 =\sup_{0\le\theta<1}
 \left\{\cE_{\mu,\theta}(R,C)
             -\frac{\theta P}{1-\theta}\right\}.
\end{equation}
Under perfect perception,
\begin{equation}\label{eq:rdpc-exact-energy}
 D_\mu(R,C,0)=\cE_{\mu,1}(R,C).
\end{equation}
\end{proposition}

It therefore suffices to compare the two-error problems at each weight:
\begin{equation}\label{eq:two-arm-target}
 \cE_{\mu,\theta}(R,C)\le\cE_{G,\theta}(R,C),
 \qquad 0\le\theta\le1.
\end{equation}
The perception tolerance no longer appears explicitly. Instead, we seek a common auxiliary variable that satisfies two information constraints and controls a weighted quadratic cost.

To identify the moments the construction should improve, center the source and write  $V=\Tr\Sigma$. For a feasible chain, define
\begin{equation}\label{eq:rdpc-moment-bookkeeping}
 q=\E\norm{\E[X-m\mid U]}^2,\qquad
 c=\E\langle X-m,Y-m\rangle.
\end{equation}
Conditional independence and orthogonality give
\begin{equation}\label{eq:weighted-energy-moments}
 d_0=V-q,\qquad d_1=V+q-2c,\qquad
 d_0+\theta d_1=(1+\theta)V-(1-\theta)q-2\theta c.
\end{equation}
Thus, increasing the posterior energy $q$ and the endpoint correlation $c$ improves each weighted objective.

\subsubsection{The Gaussian moment targets}
\label{subsec:gaussian-moment-targets}

Towards a construction that attains or improves the Gaussian
benchmark \eqref{eq:gaussian-parallel}, 
we first identify the moments selected by the Gaussian weighted problem. By translation and orthogonal invariance, we may work in a covariance eigenbasis. Specifically, for a general $\Sigma\succ0$, write $\Sigma=O\Lambda O^\top$, where $O$ is orthogonal and $\Lambda=\diag(\sigma_1^2,\ldots,\sigma_d^2)\succ0$, and define
\[
 \overline X=O^\top(X-m),\qquad
 \overline Y=O^\top(Y-m),\qquad
 \overline\mu=\Law(\overline X).
\]
This affine isometry preserves squared Euclidean distances, quadratic Wasserstein distances, and both mutual-information constraints, while leaving the auxiliary variable $U$ unchanged. It also preserves the Markov property. Moreover,
\[
 \E[\overline X\mid U]=O^\top\bigl(\E[X\mid U]-m\bigr),
\]
so the errors $d_0$ and $d_1$ are unchanged. The inverse transformation gives the converse correspondence; hence
$D_{\overline\mu}(R,C,P)=D_\mu(R,C,P)$ and
$\cE_{\overline\mu,\theta}(R,C)=\cE_{\mu,\theta}(R,C)$.
Thus, the following lemma and marked bridge construction can be formulated for a centered source in a covariance eigenbasis. In these coordinates, we continue to write $\mu$ for the transformed law and $\Sigma$ for its diagonal covariance $\Lambda$; the observation matrices $A$ and $B$ are expressed in the same basis. The conclusions for the original source follow by applying the inverse map $z\mapsto m+Oz$.

In these coordinates, a Gaussian optimizer can be represented using two diagonal matrices of observation strengths. Parameterize each scalar strength by $\rho(r)=\sqrt{1-e^{-2r}}$. The corresponding Gaussian information costs are expressed using its inverse $g=\rho^{-1}$ and the sum $j$ of the coordinate costs:
\begin{equation}\label{eq:g-j-definition}
 \begin{aligned}
 g(a)&=-\frac12\log(1-a^2),\qquad 0\le a<1,\\
 j(A)&=\sum_{i=1}^d g(a_i)=-\frac12\log\det(I_d-A^2),
 \qquad A=\diag(a_i).
 \end{aligned}
\end{equation}

\begin{lemma}[Gaussian representation and ordered strengths]
\label{lem:gaussian-normal-form}
In the covariance coordinates above, let $G=\mcN(0,\Sigma)$ with
$\Sigma=\diag(\sigma_1^2,\ldots,\sigma_d^2)\succ0$, and let $R,C\ge0$
be finite. For
$0\le\theta\le1$,
\begin{equation}\label{eq:gaussian-weighted-energy}
 \cE_{G,\theta}(R,C)
 =(1+\theta)V-
 \max_{\substack{0\preceq A,B\prec I_d\text{ diagonal}\\
                 j(A)\le R,\ j(B)\le R+C}}
 \sum_{i=1}^d\sigma_i^2
      \bigl[(1-\theta)a_i^2+2\theta a_ib_i\bigr].
\end{equation}
An optimizer can be chosen with $A\preceq B$. If $R>0$ and $\theta>0$, every optimizer satisfies $A\preceq B$, saturates both information constraints, and has matching supports: $a_i>0$ if and only if $b_i>0$. If $\theta=0$, one may take $B=A$.
\end{lemma}

The proof in \Cref{app:gaussian-structure} first bounds the conditional means of an arbitrary chain with Gaussian endpoints in terms of its coordinatewise information costs. Independent Gaussian channels sharing a common Gaussian observation attain these bounds. The optimality conditions then imply an ordering of the observation strengths. In particular, the matrices $A$ and $B$ describe a Gaussian optimizer for \eqref{eq:weighted-energy}.

For these strengths, the two Gaussian moments are
\begin{equation}\label{eq:rdpc-gaussian-targets}
 q_G=\Tr(\Sigma A^2),\qquad c_G=\Tr(\Sigma AB).
\end{equation}
We will construct a chain for the original source with information costs at most $j(A)$ and $j(B)$, posterior energy at least $q_G$, and endpoint correlation at least $c_G$. The ordering $A\preceq B$ suggests placing the common auxiliary closer to one endpoint of a midpoint bridge, allowing the two endpoints to have different information costs with respect to the same auxiliary. This motivates the following marked midpoint construction.

\subsection{The marked midpoint construction}
\label{subsec:rdpc-construction}

To construct such a chain, we retain an intermediate observation along one half of the midpoint bridge and use it as the auxiliary variable $U$. One endpoint is linked to $U$ through a stronger Gaussian observation, while the other reaches it through the midpoint. We choose the observation strengths to meet the two information budgets while preserving the required moment gains. We state this extension of the midpoint construction in the centered covariance coordinates established in \Cref{subsec:gaussian-moment-targets}.

\begin{theorem}[Marked midpoint coupling]
\label{thm:marked-bridge}
Let $\mu$ be a centered probability law on $\R^d$ with covariance
$\Sigma=\diag(\sigma_i^2)\succ0$.  For every pair of diagonal matrices
$0\preceq A\preceq B\prec I_d$, there exists a Markov chain $X-U-Y$,
with $X,Y\sim\mu$, satisfying
\begin{align}
 I(X;U)&\le j(A),& I(Y;U)&\le j(B),
 \label{eq:marked-theorem-rates}\\
 \E\norm{\E[X\mid U]}^2&\ge\Tr(\Sigma A^2),&
 \E\ip{X}{Y}&\ge\Tr(\Sigma AB).
 \label{eq:marked-theorem-moments}
\end{align}
\end{theorem}

We first specify the joint law, leaving the observation strengths free. We then bound the moments and mutual information costs in terms of these strengths. These bounds guide the choice of parameters through the Poincar\'e--Miranda theorem.

\subsubsection{Gluing the two half-bridges}

Whiten the source and denote the resulting standardized law by $\nu$:
\[
 W=\Sigma^{-1/2}X\sim\nu,\qquad \E W=0,\quad\Cov(W)=I_d.
\]
We use the reversible Gaussian observation kernels $K_D$ defined in
\eqref{eq:OU-kernel-definition}, with $\gamma=\mcN(0,I_d)$ and diagonal
$0\preceq D\preceq I_d$. These kernels obey the composition rule
\eqref{eq:OU-semigroup}.

Start with the stronger observation $\widetilde W\xrightarrow{K_B}U$, where $\widetilde W\sim\nu$, and degrade $U$ to a coarser observation $Z$ through $K_\Delta$. The composite observation has strength $T=B\Delta$. Thus, we can form a second half-bridge, $W\xrightarrow{K_T}Z$, with the same endpoint marginal:
\[
 \nu K_T=\nu K_BK_\Delta.
\]
This shared marginal is where the two half-bridges meet.

The lower-rate Gaussian reference path
$W\xrightarrow{K_T}Z\xrightarrow{K_\Delta}U$ has composite strength $T\Delta$. For a Gaussian source, the reverse transition from $Z$ to $U$ is again $K_\Delta$, so choosing $T\Delta=A$ gives $I(W;U)=j(A)$. For a general source, however, this reverse transition depends on $\nu$, and the Gaussian channel formula cannot be applied directly to $I(W;U)$.
We therefore leave the composite strength adjustable, writing
$A_\alpha=\diag(\alpha_i)$ with $0\le\alpha_i\le a_i$. Reducing the strength creates a rate margin of $j(A)-j(A_\alpha)$.
We leave $\alpha$ unspecified until we have bounded the actual information cost and both moments, so that its choice can account for all the constraints. The two composition requirements are
\begin{equation}\label{eq:marked-design-equations}
 B\Delta=T,\qquad T\Delta=A_\alpha.
\end{equation}
They determine
\begin{equation}\label{eq:marked-design-parameters}
 t_i=\sqrt{b_i\alpha_i},\qquad
 \delta_i=\sqrt{\alpha_i/b_i}\quad\text{if }b_i>0;
 \qquad t_i=\delta_i=0\quad\text{if }b_i=0.
\end{equation}
Thus, $T^2=BA_\alpha$ and $0\preceq T\preceq B\prec I_d$. The mark $U$ lies between $\widetilde W$ and the midpoint $Z$. When $A_\alpha=B$, we can take $\Delta=I_d$, so $U=Z$ and the ordinary midpoint construction is recovered.

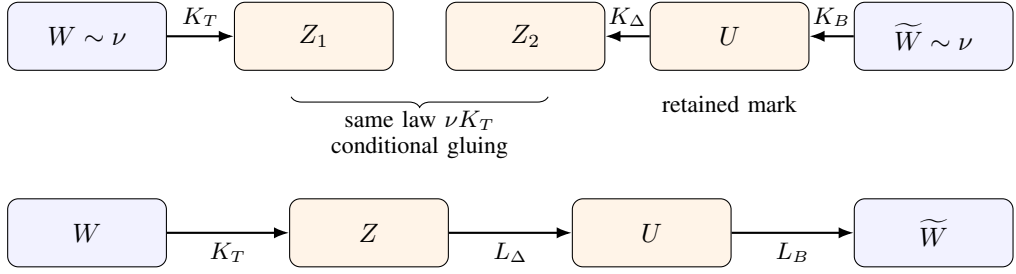
\begin{figure}[t]
\centering
\begin{tikzpicture}[
 endpoint/.style={draw,rounded corners,minimum width=21mm,minimum height=9mm,
                   fill=blue!5,align=center},
 interior/.style={endpoint,fill=orange!9},
 channel/.style={-{Latex[length=2mm]},thick},
 lab/.style={font=\small,align=center},
 node distance=19mm]
 \node[endpoint] (w) at (0,0) {$W\sim\nu$};
 \node[interior] (zl) at (3,0) {$Z_1$};
 \node[interior] (zr) at (5.8,0) {$Z_2$};
 \node[interior] (u) at (8.5,0) {$U$};
 \node[endpoint] (y) at (11.2,0) {$\widetilde W\sim\nu$};
 \draw[channel] (w) -- node[above,lab] {$K_T$} (zl);
 \draw[channel] (y) -- node[above,lab] {$K_B$} (u);
 \draw[channel] (u) -- node[above,lab] {$K_\Delta$} (zr);
 \draw[decorate,decoration={brace,mirror,amplitude=4pt}]
 (2.7,-0.8) -- (6.1,-0.8);
 \node[lab] at (4.4,-1.3) {same law $\nu K_T$\\conditional gluing};
 \node[lab,below=2mm of u] {retained mark};
 \node[endpoint] (w2) at (0,-2.6) {$W$};
 \node[interior] (z2) at (3.73,-2.6) {$Z$};
 \node[interior] (u2) at (7.47,-2.6) {$U$};
 \node[endpoint] (y2) at (11.2,-2.6) {$\widetilde W$};
 \draw[channel] (w2) -- node[below,lab] {$K_T$} (z2);
 \draw[channel] (z2) -- node[below,lab] {$L_\Delta$} (u2);
 \draw[channel] (u2) -- node[below,lab] {$L_B$} (y2);
\end{tikzpicture}
\caption{The top row shows two observation laws with the same midpoint marginal, $\Law(Z_1)=\Law(Z_2)=\nu K_T$. Conditional gluing produces the path shown in the bottom row. Its transitions $L_\Delta$ and $L_B$ are the Bayes reversals of the observation steps $U\to Z_2$ and $\widetilde W\to U$, respectively, under the source-dependent marginals. The resulting variables form the Markov chain $W-Z-U-\widetilde W$.}
\label{fig:marked-midpoint}
\end{figure}

Fig. \ref{fig:marked-midpoint} distinguishes the observation structure from the sampling procedure. The top row shows the forward Gaussian observations that define the two half-bridges. In the bottom row, we sample $W$, then $Z$, and then the right half-bridge conditional on $Z$. The transition kernels are Bayes reversals under the corresponding source-dependent marginals.

We now describe the joint law explicitly. Define
\begin{equation}\label{eq:marked-smoothed-laws}
 \nu_B=\nu K_B,\quad \nu_T=\nu K_T,\qquad
 f_B=\frac{\dd \nu_B}{\dd \gamma},\quad
 f_T=\frac{\dd \nu_T}{\dd \gamma}.
\end{equation}
Since $B,T\prec I_d$, both observations have nondegenerate Gaussian noise. Hence the smoothed densities $f_B$ and $f_T$ exist and are strictly positive, even if $\nu$ has no density. In our kernel notation, the first argument is the input, and the second denotes the output measure. Thus, the forward observation $U\to Z$ has joint law $\nu_B(\dd u)K_\Delta(u,\dd z)$. Reversibility with respect to $\gamma$ gives
\[
 \begin{split}
 \nu_B(\dd u)K_\Delta(u,\dd z)
 &=f_B(u)\gamma(\dd u)K_\Delta(u,\dd z)\\
 &=\gamma(\dd z)f_B(u)K_\Delta(z,\dd u).
 \end{split}
\]
The $Z$ marginal is $\nu_T$, since $\nu_BK_\Delta=\nu_T$. Integrating the last expression over $u$ therefore yields
\[
 f_T(z)=\int f_B(u)K_\Delta(z,\dd u)=(K_\Delta f_B)(z).
\]
Dividing by this marginal density gives the conditional law of $U$ given $Z=z$:
\begin{equation}\label{eq:marked-reverse-kernel}
 L_\Delta(z,\dd u)=\frac{f_B(u)}{f_T(z)}K_\Delta(z,\dd u).
\end{equation}
Indeed,
\[
 \nu_T(\dd z)L_\Delta(z,\dd u)=\nu_B(\dd u)K_\Delta(u,\dd z),\qquad
 \int L_\Delta(z,\dd u)=\frac{(K_\Delta f_B)(z)}{f_T(z)}=1.
\]
The resulting law
\begin{equation}\label{eq:marked-path-law}
 \mathsf P_{WZU}(\dd w,\dd z,\dd u)
 =\nu(\dd w)K_T(w,\dd z)L_\Delta(z,\dd u)
\end{equation}
has marginals $\nu$, $\nu_T$, and $\nu_B$ for $W$, $Z$, and $U$, respectively, and satisfies the Markov property $W-Z-U$. We write $\mathsf P_{WU}$ and $\mathsf P_U$ for its corresponding marginals.
To complete the path, disintegrate the stronger observation:
\begin{equation}\label{eq:marked-high-disintegration}
 \nu(\dd \widetilde w)K_B(\widetilde w,\dd u)
 =\nu_B(\dd u)L_B(u,\dd \widetilde w).
\end{equation}
This posterior kernel exists for any source law $\nu$. Indeed, since $B\prec I_d$, the Gaussian kernel $K_B(\widetilde w,\dd u)$ has a Lebesgue density, which we denote by $k_B(u\mid\widetilde w)$. Bayes' formula then gives
\[
 L_B(u,\dd \widetilde w)
 =\frac{k_B(u\mid\widetilde w)\nu(\dd \widetilde w)}
        {\int k_B(u\mid v)\nu(\dd v)}.
\]
The denominator is strictly positive, so this formula does not require $\nu$ to have a density.
Extend the path by sampling $\widetilde W$ from $L_B(U,\cdot)$, conditionally independently of $(W,Z)$ given $U$. Then
\begin{equation}\label{eq:marked-full-chain}
 W-Z-U-\widetilde W,\qquad W,\widetilde W\sim\nu.
\end{equation}
The two reverse-kernel identities show that $(\widetilde W,U,Z)$ has joint law
\[
 \nu(\dd \widetilde w)K_B(\widetilde w,\dd u)K_\Delta(u,\dd z).
\]
Since $K_BK_\Delta=K_T$, the pairs $(\widetilde W,Z)$ and $(W,Z)$ have the same law. The Markov property then implies that, conditional on $Z$, the two endpoints are independent draws from the same posterior distribution. Thus, retaining the mark $U$ preserves the midpoint coupling associated with $K_T$.

\subsubsection{Posterior energy and endpoint correlation}
\label{subsec:marked-estimates}

For the constructed chain, we seek lower bounds on
\[
 \E\norm{\E[W\mid U]}^2
 \qquad\text{and}\qquad
 \E\ip{W}{\widetilde W}.
\]
We keep the estimates coordinatewise so that the covariance weights
$\sigma_i^2$ can be restored at the end. The midpoint $Z$ makes the
second quantity directly accessible: given $Z$, the endpoints are
independent and have the same conditional mean
\[
 m_T(Z):=\E[W\mid Z]=\E[\widetilde W\mid Z].
\]
Therefore,
\[
 \E[W\widetilde W^\top]
 =\E\bigl[\E[W\mid Z]\E[\widetilde W\mid Z]^\top\bigr]
 =\E[m_T(Z)m_T(Z)^\top].
\]
Thus, the endpoint cross-moment is determined by the second moment of
the posterior mean at the midpoint. To track this moment as the
observation strength varies, define
\begin{equation}\label{eq:marked-posterior-energy}
 Q_\nu(D)
 =\E\bigl[\E[W\mid Z_D]\E[W\mid Z_D]^\top\bigr],\qquad
 Z_D=DW+(I_d-D^2)^{1/2}G,
\end{equation}
where $G\sim\gamma$ is independent of $W$. For a Gaussian source, the posterior mean is $DZ_D$, so $Q_\gamma(D)=D^2$. For a general standardized source, $DZ_D$ remains the best linear estimate, but the posterior mean may also have a nonlinear component. Since
$\E[Z_DZ_D^\top]=I_d$ and
$\E[\E[W\mid Z_D]Z_D^\top]=\E[WZ_D^\top]=D$, expanding the square gives
\begin{equation}\label{eq:marked-residual-matrix}
 Q_\nu(D)-D^2
 =\E\bigl[(\E[W\mid Z_D]-DZ_D)
          (\E[W\mid Z_D]-DZ_D)^\top\bigr]\succeq0.
\end{equation}
At the midpoint, define the excess in each coordinate by
\begin{equation}\label{eq:marked-excess-definition}
 e_i:=\E\bigl[(m_{T,i}(Z)-t_iZ_i)^2\bigr]
     =Q_\nu(T)_{ii}-t_i^2\ge0.
\end{equation}
Since $t_i^2=b_i\alpha_i$, the endpoint cross-moment satisfies
\begin{equation}\label{eq:marked-endpoint-energy}
 \E[W\widetilde W^\top]=Q_\nu(T),\qquad
 \E[W_i\widetilde W_i]=b_i\alpha_i+e_i.
\end{equation}
Here, $b_i\alpha_i$ is the Gaussian value at the chosen strengths, and $e_i$ is the additional endpoint correlation.

To bound the other target, define
$\overline W=\E[W\mid U]$ and $H=\E[WU^\top]$.
The linear predictor $HU$ provides a lower bound on the second moment of $\overline W$. Indeed, since $\E[UU^\top]=I_d$ and $\E[\overline WU^\top]=H$,
\begin{equation}\label{eq:marked-posterior-linear-bound}
 \E[\overline W\overline W^\top]-HH^\top
 =\E[(\overline W-HU)(\overline W-HU)^\top]\succeq0.
\end{equation}
In particular,
\[
 \E\norm{\E[W\mid U]}^2
 \ge\Tr(HH^\top)\ge\sum_i H_{ii}^2.
\]
It remains to compute $H_{ii}=\E[W_iU_i]$. The midpoint again helps: $W$ and $U$ are conditionally independent given $Z$, so
\[
 H_{ii}=\E\bigl[m_{T,i}(Z)\E[U_i\mid Z]\bigr].
\]
We compute $\E[U_i\mid Z]$ using the forward Gaussian chain
$\widetilde W\to U\to Z$. Conditional on $\widetilde W=y$,
\[
 U_i=b_i y_i+\sqrt{1-b_i^2}\,G_i,\qquad
 Z_i=\delta_i U_i+\sqrt{1-\delta_i^2}\,G_i',
\]
where $G_i$ and $G_i'$ are independent standard Gaussian variables. Conditional on $y$, the pair $(U_i,Z_i)$ has means $(b_i y_i,t_i y_i)$ and covariance entries
\[
 \Var(U_i\mid y)=1-b_i^2,\quad
 \Var(Z_i\mid y)=1-t_i^2,\quad
 \Cov(U_i,Z_i\mid y)=\delta_i(1-b_i^2).
\]
Gaussian regression therefore gives
\begin{equation}\label{eq:marked-Gaussian-regression}
 \E[U_i\mid\widetilde W=y,Z=z]
 =b_i y_i+
 \frac{\delta_i(1-b_i^2)}{1-t_i^2}(z_i-t_i y_i).
\end{equation}
Set $\beta_i=\delta_i(1-b_i^2)/(1-t_i^2)$. Averaging
\eqref{eq:marked-Gaussian-regression} over the posterior distribution of
$\widetilde W$ given $Z$, whose mean is $m_T(Z)$, yields
\[
 \E[U_i\mid Z]=k_i m_{T,i}(Z)+\beta_i Z_i,\qquad
 k_i:=b_i-\beta_i t_i
     =\frac{b_i-\alpha_i}{1-b_i\alpha_i}\ge0.
\]
Substituting this expression into the formula for $H_{ii}$ and using
$\E[m_{T,i}(Z)Z_i]=\E[W_iZ_i]=t_i$, we obtain
\begin{align}
 H_{ii}
 &=k_iQ_\nu(T)_{ii}+\beta_i t_i\notag\\
 &=\alpha_i+k_i e_i.
 \label{eq:marked-H-identity}
\end{align}
Here, the last equality follows from
$Q_\nu(T)_{ii}=t_i^2+e_i$ and
$k_i t_i^2+\beta_i t_i=\alpha_i$.
Thus, the same midpoint gains $e_i$ improve both quantities of interest:
\[
 \begin{aligned}
 \E\norm{\E[W\mid U]}^2&\ge\sum_i(\alpha_i+k_i e_i)^2,\\
 \E\ip{W}{\widetilde W}&=\sum_i(b_i\alpha_i+e_i).
 \end{aligned}
\]
Choosing $\alpha_i=a_i$ meets the Gaussian moment targets. To satisfy the information constraints as well, we next determine how these gains affect the bound on $I(W;U)$. This will guide our choice of $\alpha_i\le a_i$.

\subsubsection{Information bounds}

The pair $(\widetilde W,U)$ is generated by a Gaussian observation kernel, so its information bound follows directly from the Gaussian channel identity. The pair $(W,U)$, however, involves the reverse kernel $L_\Delta$. We therefore compare the path through $Z$ with a Gaussian reference path.


For the first bound, recall that a standardized input with law $\nu$, observed through $K_D$ for $0\preceq D\prec I_d$, satisfies
\begin{equation}\label{eq:marked-Gaussian-channel-accounting}
 I(W;Z_D)+\KL{\nu K_D}{\gamma}=j(D).
\end{equation}
To see this, apply the relative-entropy chain rule to the joint law
$\nu(\dd w)K_D(w,\dd z)$ with respect to $\nu(\dd w)\gamma(\dd z)$:
\[
 \E \KL{K_D(W,\cdot)}{\gamma}
 =I(W;Z_D)+\KL{\nu K_D}{\gamma}.
\]
Since $\E[WW^\top]=I_d$, the left-hand side is
$-\frac12\log\det(I_d-D^2)=j(D)$. Because the two joint laws have the same input marginal $\nu$, no density for the source is needed; the identity also shows that the smoothed relative entropy is finite. Applying this identity to the observation $\widetilde W\xrightarrow{K_B}U$ gives
\begin{equation}\label{eq:marked-high-information}
 I(\widetilde W;U)=j(B)-\KL{\nu_B}{\gamma}\le j(B).
\end{equation}
Thus, this information constraint holds for every choice of $\alpha$.

For $I(W;U)$, the conditional law of $U$ given $W$ is generally not a
Gaussian channel. To bound it, compare the path law in
\eqref{eq:marked-path-law} with the reference path obtained by replacing
$L_\Delta$ with $K_\Delta$:
\begin{equation}\label{eq:marked-reference-path}
 \mathsf Q_{WZU}(\dd w,\dd z,\dd u)
 =\nu(\dd w)K_T(w,\dd z)K_\Delta(z,\dd u).
\end{equation}
The endpoint channel of this reference path is
$K_TK_\Delta=K_{A_\alpha}$. Retaining $Z$ makes the density ratio
explicit: the two path laws share the factor
$\nu(\dd w)K_T(w,\dd z)$, and \eqref{eq:marked-reverse-kernel} gives
\[
 \frac{\dd \mathsf P_{WZU}}{\dd \mathsf Q_{WZU}}(w,z,u)
 =\frac{f_B(u)}{f_T(z)}.
\]
Taking logarithms and integrating with respect to $\mathsf P_{WZU}$
yields
\begin{equation}\label{eq:marked-path-KL}
 \begin{split}
 \KL{\mathsf P_{WZU}}{\mathsf Q_{WZU}}
 &=\E_{\mathsf P_{WZU}}[\log f_B(U)-\log f_T(Z)]\\
 &=\KL{\nu_B}{\gamma}-\KL{\nu_T}{\gamma}.
 \end{split}
\end{equation}
Projecting onto $(W,U)$ then bounds the divergence from the Gaussian
channel:
\begin{equation}\label{eq:marked-projected-KL}
 \KL{\mathsf P_{WU}}{\nu(\dd w)K_{A_\alpha}(w,\dd u)}
 \le \KL{\nu_B}{\gamma}-\KL{\nu_T}{\gamma}.
\end{equation}
To obtain a mutual-information bound, change the reference measure to
the product $\nu\otimes\nu_B$ of the actual marginals. The term
$\KL{\nu_B}{\gamma}$ then cancels:
\begin{align}
 I(W;U)
 &=\KL{\mathsf P_{WU}}{\nu(\dd w)K_{A_\alpha}(w,\dd u)}
   -\KL{\nu_B}{\gamma}\notag+\E\log\frac{\dd K_{A_\alpha}(W,\cdot)}{\dd \gamma}(U)\notag\\
 &\le-\KL{\nu_T}{\gamma}
   +\E\log\frac{\dd K_{A_\alpha}(W,\cdot)}{\dd \gamma}(U).
 \label{eq:marked-information-reference-change}
\end{align}
The remaining log-likelihood ratio is quadratic, so its expectation
depends only on the second moments already computed:
\begin{equation}\label{eq:marked-Gaussian-log-density}
 \log\frac{\dd K_{A_\alpha}(w,\cdot)}{\dd \gamma}(u)
 =j(A_\alpha)+\sum_{i=1}^d
 \left\{\frac{\alpha_iw_iu_i}{1-\alpha_i^2}
 -\frac{\alpha_i^2(w_i^2+u_i^2)}{2(1-\alpha_i^2)}\right\}.
\end{equation}
Both $W$ and $U$ are centered with covariance $I_d$, so the expected quadratic terms involving
 $W^2_i$ and $U^2_i$ match their Gaussian values. The correction
therefore depends on $H_{ii}-\alpha_i=k_i e_i$. Taking expectations and
using $g'(\alpha)=\alpha/(1-\alpha^2)$ gives
\begin{align}
 I(W;U)
 &\le j(A_\alpha)-\KL{\nu_T}{\gamma}
       +\sum_i g'(\alpha_i)(H_{ii}-\alpha_i)\notag\\
 &=\sum_i\bigl[g(\alpha_i)+g'(\alpha_i)k_i e_i\bigr]
       -\KL{\nu_T}{\gamma}.
 \label{eq:marked-low-information}
\end{align}
This estimate explains why we left the strengths adjustable. At the
unadjusted choice $\alpha_i=a_i$, it becomes
\[
 I(W;U)\le j(A)+\sum_i g'(a_i)k_i e_i-\KL{\nu_T}{\gamma},
\]
with $T$, $k_i$, and $e_i$ evaluated at that choice. Although the
moment estimates give $e_i\ge0$, they do not show that the nonnegative
correction is bounded by $\KL{\nu_T}{\gamma}$. Thus they do not by
themselves establish $I(W;U)\le j(A)$. For a Gaussian source, both
terms vanish, so the unadjusted choice $A_\alpha=A$ suffices.

Choosing $\alpha_i\le a_i$ creates a rate margin
$g(a_i)-g(\alpha_i)$ that can offset the correction
$g'(\alpha_i)k_i e_i$. Reducing the observation strengths also lowers
the baseline moments, but the gains $e_i$ and $k_i e_i$ can compensate
for this loss. We therefore seek coordinatewise choices of $\alpha_i$
that balance the rate and moment requirements. The next subsection
shows that such a choice exists and meets both moment targets.

\subsection{Rate balancing and Gaussian extremality}
\label{subsec:rdpc-rate-selection}

The preceding estimates reduce the construction to choosing the
observation strengths. We first choose them to satisfy both information
constraints, then use the resulting moment bounds to compare the RDPC
objective with the Gaussian benchmark in \eqref{eq:gaussian-parallel}.

\subsubsection{Simultaneous rate balancing}

\begin{proof}[Completion of the proof of \Cref{thm:marked-bridge}]
Let $\mathcal I=\{i:a_i>0\}$. If $\mathcal I=\varnothing$, take
independent endpoints and a constant auxiliary. Otherwise, define
\[
 \mathcal C=\prod_{i\in\mathcal I}[0,a_i]
\]
and set $\alpha_i=0$ for $i\notin\mathcal I$. Each
$\alpha\in\mathcal C$ determines a construction. Choose the observation
strengths by solving
\begin{equation}\label{eq:rate-stop}
 F_i(\alpha):=g(\alpha_i)+g'(\alpha_i)k_i(\alpha)e_i(\alpha)-g(a_i)=0,
 \qquad i\in\mathcal I.
\end{equation}
This equation balances the rate saved,
$g(a_i)-g(\alpha_i)$, against the posterior correction
$g'(\alpha_i)k_i e_i$. Although the system has one equation per
coordinate, each posterior correction depends on the full observation
vector. Dependence among the source coordinates therefore couples the
equations.

The posterior-energy continuity lemma in
\Cref{app:posterior-continuity} shows that each $e_i(\alpha)$ is jointly
continuous on $\mathcal C$. Since $T(\alpha)\preceq B\prec I_d$, the
observation noise remains uniformly nondegenerate, and the other terms
in $F_i$ are continuous as well. On each pair of opposite faces,
\begin{equation}\label{eq:marked-PM-faces}
 F_i(\alpha)\big|_{\alpha_i=0}=-g(a_i)<0,\qquad
 F_i(\alpha)\big|_{\alpha_i=a_i}=g'(a_i)k_i e_i\ge0.
\end{equation}
The Poincar\'e--Miranda theorem \cite{Miranda1940,Kulpa1997} 
therefore
gives a common zero $\alpha^*\in\mathcal C$. As a multidimensional
analogue of the intermediate value theorem, it ensures that all the
rate-balance equations vanish at the same point. The strict sign on
the lower face implies $\alpha_i^*>0$ for every $i\in\mathcal I$.

Use the common zero to set the strengths in the marked bridge. Summing
\eqref{eq:rate-stop} and applying \eqref{eq:marked-low-information} gives
\begin{equation}\label{eq:marked-low-information-stopped}
 I(W;U)\le j(A)-\KL{\nu_T}{\gamma}\le j(A).
\end{equation}
The other information constraint follows from
\eqref{eq:marked-high-information} .


It remains to verify that the rate adjustment preserves the moment
targets. The rate balances ensure this. By convexity of $g$, for each
active coordinate,
\begin{equation}\label{eq:marked-convex-rate-balance}
 g'(\alpha_i)k_i e_i
 =g(a_i)-g(\alpha_i)
 \ge g'(\alpha_i)(a_i-\alpha_i).
\end{equation}
Since $g'(\alpha_i)>0$, we have $k_i e_i\ge a_i-\alpha_i$, and hence
\begin{equation}\label{eq:marked-H-target}
 H_{ii}=\alpha_i+k_i e_i\ge a_i.
\end{equation}
If $k_i>0$, then
\[
 \frac{1}{k_i}-b_i=\frac{1-b_i^2}{b_i-\alpha_i}>0,
\]
so $e_i\ge b_i(a_i-\alpha_i)$. If $k_i=0$, the ordering
$\alpha_i\le a_i\le b_i$ forces $\alpha_i=a_i=b_i$. Thus, in either
case,
\begin{equation}\label{eq:marked-Q-target}
 Q_\nu(T)_{ii}=b_i\alpha_i+e_i\ge a_i b_i.
\end{equation}
For inactive coordinates, $a_i=\alpha_i=0$, so
$H_{ii}=b_i e_i\ge0$ and $Q_\nu(T)_{ii}\ge0$.

Returning to the original covariance coordinates in
\eqref{eq:marked-posterior-linear-bound} and
\eqref{eq:marked-endpoint-energy}  gives
\begin{align}
 \E\norm{\Sigma^{1/2}\overline W}^2
 &\ge\Tr(\Sigma HH^\top)
 \ge\sum_i\sigma_i^2 H_{ii}^2
 \ge\Tr(\Sigma A^2),
 \label{eq:marked-q-weighted}\\
 \E\ip{\Sigma^{1/2}W}{\Sigma^{1/2}\widetilde W}
 &=\Tr(\Sigma Q_\nu(T))
 \ge\Tr(\Sigma AB).
 \label{eq:marked-c-weighted}
\end{align}
The endpoints $X=\Sigma^{1/2}W$ and
$Y=\Sigma^{1/2}\widetilde W$ both have law $\mu$. This invertible
change of coordinates preserves mutual information and the Markov
property, while \eqref{eq:marked-q-weighted}--\eqref{eq:marked-c-weighted}
give the required Euclidean moment bounds.
\end{proof}

The construction, moment estimates, information bounds, and posterior-energy
continuity lemma require only finite second moments; no density for the
source is needed. Applying the inverse isometry from
\Cref{subsec:gaussian-moment-targets} carries the constructed endpoints
back to the original source, with mean $m$ and covariance
$\Sigma\succ0$. This preserves both information bounds and the centered
moment bounds. In the original coordinates, the observation-strength
matrices are $OAO^\top$ and $OBO^\top$.

\begin{corollary}[Strict improvement for positive observation strengths]
\label{cor:marked-bridge-strict}
Under the assumptions of \Cref{thm:marked-bridge}, suppose that $\mu$
is non-Gaussian and $0\prec A\preceq B\prec I_d$. Then there exists a
chain $X-U-Y$ satisfying \eqref{eq:marked-theorem-rates} and strictly
exceeding both moment bounds:
\begin{equation}\label{eq:marked-strict-moments}
 \E\norm{\E[X\mid U]}^2>\Tr(\Sigma A^2),
 \qquad
 \E\ip{X}{Y}>\Tr(\Sigma AB).
\end{equation}
\end{corollary}

\begin{proof}
We improve both moments by slightly increasing the target observation
strengths and use the resulting entropy gains to absorb the additional
information cost. 
For $0<\varepsilon<1-\max_i b_i$, set
\[
 A_\varepsilon=A+\varepsilon I_d,\qquad
 B_\varepsilon=B+\varepsilon I_d,
\]
and apply the rate-selection construction to these new targets. Denote
the resulting parameters by
$\alpha_i^\varepsilon,k_i^\varepsilon,e_i^\varepsilon$, and
$T_\varepsilon$.  By
\eqref{eq:marked-low-information-stopped} and
\eqref{eq:marked-high-information}, the constructed chain satisfies
\[
 \begin{aligned}
 I(X_\varepsilon;U_\varepsilon)
 &\le j(A_\varepsilon)-\KL{\nu K_{T_\varepsilon}}{\gamma},\\
 I(Y_\varepsilon;U_\varepsilon)
 &=j(B_\varepsilon)-\KL{\nu K_{B_\varepsilon}}{\gamma}.
 \end{aligned}
\]
We now show that both relative-entropy terms stay bounded away from
zero, so that a sufficiently small increase in the target strengths
keeps the information costs within the original budgets.

The rate balances keep the selected strengths uniformly bounded away
from zero. Indeed,
$Q_\nu(T_\varepsilon)_{ii}\le \E W_i^2=1$, so
$0\le e_i^\varepsilon\le1$. Also,
\[
 0\le k_i^\varepsilon
 =\frac{b_i+\varepsilon-\alpha_i^\varepsilon}
        {1-(b_i+\varepsilon)\alpha_i^\varepsilon}
 \le b_i+\varepsilon<1.
\]
Choose $0<\tau<\min_i a_i$ small enough that
$g(\tau)+g'(\tau)<\min_i g(a_i)$; this is possible because
$g(0)=g'(0)=0$. The rate-balance equation
\eqref{eq:rate-stop} for the perturbed targets gives
\[
 g(a_i+\varepsilon)
 =g(\alpha_i^\varepsilon)
   +g'(\alpha_i^\varepsilon)k_i^\varepsilon e_i^\varepsilon
 \le g(\alpha_i^\varepsilon)+g'(\alpha_i^\varepsilon).
\]
Since $g+g'$ is increasing, $\alpha_i^\varepsilon\le\tau$ would
contradict the choice of $\tau$. Hence $\alpha_i^\varepsilon>\tau$ and
$t_i^\varepsilon=\sqrt{(b_i+\varepsilon)\alpha_i^\varepsilon}>\tau$.
Thus $T_\varepsilon,B_\varepsilon\succeq\tau I_d$.

Set $\kappa=\KL{\nu K_{\tau I_d}}{\gamma}$. This divergence is finite
by \eqref{eq:marked-Gaussian-channel-accounting}. It is strictly
positive by the characteristic-function argument in the proof of
\Cref{thm:matrix-self-coupling}: since $\tau I_d$ is invertible and
$\nu$ is non-Gaussian, $\nu K_{\tau I_d}\ne\gamma$.

Now let $D$ be any diagonal matrix with
$\tau I_d\preceq D\prec I_d$. Since
$K_DK_{\tau D^{-1}}=K_{\tau I_d}$ and
$\gamma K_{\tau D^{-1}}=\gamma$, the data processing inequality gives
\[
 \KL{\nu K_D}{\gamma}\ge\KL{\nu K_{\tau I_d}}{\gamma}=\kappa>0.
\]
In particular, this bound applies to both
$D=T_\varepsilon$ and $D=B_\varepsilon$.

By continuity of $j$, choose $\varepsilon>0$ small enough that
$j(A_\varepsilon)-j(A)<\kappa$ and
$j(B_\varepsilon)-j(B)<\kappa$. The information bounds above then
remain within the original budgets. Meanwhile,
\eqref{eq:marked-q-weighted}--\eqref{eq:marked-c-weighted} yield
\[
 \begin{aligned}
 \E\norm{\E[X_\varepsilon\mid U_\varepsilon]}^2
 &\ge\Tr(\Sigma A_\varepsilon^2)>\Tr(\Sigma A^2),\\
 \E\ip{X_\varepsilon}{Y_\varepsilon}
 &\ge\Tr(\Sigma A_\varepsilon B_\varepsilon)>\Tr(\Sigma AB).
 \end{aligned}
\]
Thus both moment bounds are strict.
\end{proof}

\subsubsection{Comparison of the weighted errors}
\label{sec:rdpc-comparison}

We now apply the marked midpoint construction to the two-error problem
in \eqref{eq:weighted-energy}. The bridge theorem matches the information
budgets and the two moments of the Gaussian comparison chain. The
variational identity then incorporates the perception constraint.


\begin{proposition}[Comparison of the two weighted errors]
\label{prop:two-arm-comparison}
Let $\mu\in\mcP_2(\R^d)$ have mean $m$ and positive-definite
covariance $\Sigma$, and let $G=\mcN(m,\Sigma)$. For all finite
$R,C\ge0$ and $0\le\theta\le1$,
\begin{equation}\label{eq:two-arm-comparison}
 \cE_{\mu,\theta}(R,C)\le\cE_{G,\theta}(R,C).
\end{equation}
If $\mu$ is non-Gaussian and the Gaussian optimization in
\eqref{eq:gaussian-weighted-energy} admits an optimizer satisfying
$0\prec A\preceq B$, then the inequality is strict. In particular, for
a non-Gaussian scalar source with $R>0$, it is strict for every
$0\le\theta\le1$.
\end{proposition}

\begin{proof}
By the coordinate reduction in \Cref{subsec:gaussian-moment-targets},
it suffices to consider $m=0$ and
$\Sigma=\diag(\sigma_1^2,\ldots,\sigma_d^2)\succ0$. Choose an ordered
Gaussian optimizer $A\preceq B$ as in
\Cref{lem:gaussian-normal-form}. The marked midpoint theorem
(\Cref{thm:marked-bridge}) gives a chain $X-U-Y$ with
$X,Y\sim\mu$,
$I(X;U)\le j(A)\le R$, and
$I(Y;U)\le j(B)\le R+C$, whose moments satisfy
\[
 q\ge\Tr(\Sigma A^2),\qquad c\ge\Tr(\Sigma AB).
\]
This chain is feasible in \eqref{eq:weighted-energy}. Since
$1-\theta$ and $\theta$ are nonnegative, \eqref{eq:weighted-energy-moments}
gives
\begin{align*}
 \cE_{\mu,\theta}(R,C)
 &\le(1+\theta)V-(1-\theta)q-2\theta c\\
 &\le(1+\theta)V-(1-\theta)\Tr(\Sigma A^2)
                         -2\theta\Tr(\Sigma AB)\\
 &=\cE_{G,\theta}(R,C).
\end{align*}

If the Gaussian problem admits an ordered optimizer with $A\succ0$ and
$\mu$ is non-Gaussian, \Cref{cor:marked-bridge-strict} makes both
moment bounds strict. For every $\theta\in[0,1]$, the coefficients
$1-\theta$ and $2\theta$ are nonnegative and cannot both vanish.
Therefore the weighted comparison is strict for every such $\theta$.
\end{proof}

\subsubsection{Gaussian allocations and distortion comparison}
\label{subsec:rdpc-distortion-comparison}

The weighted comparison and \Cref{prop:rdpc-variational} give the
non-strict RDPC inequality. To prove strictness under the allocation
condition in \Cref{thm:rdpc-extremality}, we also apply the marked
construction to a distortion-minimizing Gaussian allocation. The
following property lets us order its observation strengths while
preserving the rate allocation for the first endpoint.


\begin{proposition}[Ordered distortion-minimizing Gaussian allocations]
\label{prop:gaussian-parallel}
For finite $R,C\ge0$ and $P\ge0$, every minimizing allocation in
\eqref{eq:gaussian-parallel} can be replaced by a minimizing allocation
with the same rates $r_i$ and with $r_i\le s_i$ for every $i$.
\end{proposition}

\begin{proof}
For the parallel Gaussian construction, define
\[
 d_0(r)=\sum_i\sigma_i^2e^{-2r_i},\qquad
 d_1(r,s)=\sum_i\sigma_i^2
    \bigl(1+\rho(r_i)^2-2\rho(r_i)\rho(s_i)\bigr).
\]
These are the two squared-error terms in the common-auxiliary
representation. By \Cref{lem:gaussian-normal-form},
\[
 \cE_{G,\theta}(R,C)
 =\min_{r,s}\{d_0(r)+\theta d_1(r,s)\},
 \qquad 0\le\theta\le1,
\]
where the minimum is taken over the rate constraints in
\eqref{eq:gaussian-parallel}.

If $R=0$, then $r_i=0$ for every $i$, and setting $s_i=0$ leaves the
objective unchanged, so an ordered minimizer exists. Now suppose $R>0$
and let $(r,s)$ minimize \eqref{eq:gaussian-parallel}. At $P=0$, the
objective is $d_0(r)+d_1(r,s)$, so
\Cref{lem:gaussian-normal-form} gives $r_i\le s_i$ for every $i$.


For $P>0$, define
\[
 \Phi(r,s,\theta)
 :=d_0(r)+\theta d_1(r,s)-\frac{\theta P}{1-\theta}.
\]
Let $\theta_*$ be a Gaussian dual maximizer, whose existence is
established in \Cref{app:rdpc-variational}. The weighted representation,
\eqref{eq:rdpc-energy-duality}, and \eqref{eq:perception-variational} imply
\[
 D_G(R,C,P)=\min_{u,v}\Phi(u,v,\theta_*)
 \le\Phi(r,s,\theta_*)
 \le d_0(r)+h_P(d_1(r,s))=D_G(R,C,P).
\]
Hence equality holds throughout, and $(r,s)$ minimizes
$d_0+\theta_*d_1$. If $\theta_*>0$, then
\Cref{lem:gaussian-normal-form} again gives $r_i\le s_i$ for every $i$.

When $\theta_*=0$, equality in the preceding display implies that
$r$ minimizes $d_0$ and $d_1(r,s)\le P$. For this fixed $r$, a maximizer
$\bar s$ of $\sum_i\sigma_i^2\rho(r_i)\rho(\bar s_i)$ over the second
rate simplex minimizes $d_1(r,\bar s)$. Consequently,
$d_1(r,\bar s)\le d_1(r,s)\le P$, and $(r,\bar s)$ remains a
distortion-minimizing allocation with the same first rates.

The strict decrease of $d_0$ in each coordinate gives
$\sum_i r_i=R$, so $\mathcal I:=\{i:r_i>0\}$ is nonempty.
For fixed $r$, the cross-moment is strictly increasing in $\bar s_i$
on $\mathcal I$, which also gives $\sum_i\bar s_i=R+C$.
Moreover, a positive second rate outside $\mathcal I$ would contribute
nothing to the cross-moment; transferring it to a coordinate in
$\mathcal I$ would strictly improve the objective. Thus
$\{i:\bar s_i>0\}\subseteq\mathcal I$.
The reverse inclusion follows from the second-rate transfer argument
in \Cref{app:gaussian-structure}, with $r$ held fixed. Indeed, if
$i\in\mathcal I$ had $\bar s_i=0$, transferring a sufficiently small
rate $\varepsilon$ from some $\bar s_j>0$ to $\bar s_i$ would yield
a gain of order $\sqrt\varepsilon$ and a loss of order $\varepsilon$,
contradicting optimality. The expansion and the loss estimate are the
same as in that argument. Hence
$\mathcal I=\{i:\bar s_i>0\}$, and both budgets are attained.

With $a_i=\rho(r_i)$ and $b_i=\rho(\bar s_i)$, stationarity of the
ordinary RD minimization over $r$ and of the cross-moment maximization
over $\bar s$ gives positive multipliers $\lambda,\eta$, each common
to all coordinates in $\mathcal I$, such that
\[
 \lambda=\sigma_i^2(1-a_i^2),\qquad
 \eta=\sigma_i^2a_i\frac{1-b_i^2}{b_i},
 \qquad i\in\mathcal I.
\]
These equations have the same algebraic form as
\eqref{eq:normal-form-kkt-a}--\eqref{eq:normal-form-kkt-b} at $\theta=0$,
with multipliers obtained here from the two separate optimizations.
By \eqref{eq:normal-form-ordering-identity}, the sign of $b_i-a_i$
is therefore the same at every coordinate in $\mathcal I$.
A negative sign would imply
$R+C=\sum_i\bar s_i<\sum_i r_i=R$, which is impossible.
It follows that $r_i\le\bar s_i$ on $\mathcal I$; outside
$\mathcal I$, both rates are zero. This proves the assertion.
\end{proof}


We can now prove \Cref{thm:rdpc-extremality}.

\begin{proof}[Proof of \Cref{thm:rdpc-extremality}]
For finite $C$ and $P>0$, substituting
\eqref{eq:two-arm-comparison}  into the variational representation
\eqref{eq:rdpc-energy-duality}  gives
\[
 \begin{split}
 D_\mu(R,C,P)
 &=\sup_{0\le\theta<1}
 \left\{\cE_{\mu,\theta}(R,C)-\frac{\theta P}{1-\theta}\right\}\\
 &\le\sup_{0\le\theta<1}
 \left\{\cE_{G,\theta}(R,C)-\frac{\theta P}{1-\theta}\right\}
 =D_G(R,C,P).
 \end{split}
\]
For $P=0$, apply the weighted comparison at $\theta=1$ and use
 \eqref{eq:rdpc-exact-energy}.

 For strictness at finite $C$, choose a minimizing Gaussian allocation
in \eqref{eq:gaussian-parallel} with $r_i>0$ for every $i$. By the
ordering assertion of \Cref{prop:gaussian-parallel}, we can keep these
rates and choose $s_i\ge r_i$. In a covariance eigenbasis, the
corresponding observation strengths satisfy
$0\prec A\preceq B\prec I_d$.
By \Cref{cor:marked-bridge-strict}, there is a source chain satisfying
the information constraints whose moments obey $q>q_G$ and $c>c_G$,
where $q_G,c_G$ are the moments of the chosen Gaussian minimizer. Set
\[
 \eta=\min\{q-q_G,\,2(c-c_G)\}>0.
\]
For every $0\le\theta\le1$, the moment identity
\eqref{eq:weighted-energy-moments} gives
\[
 d_0+\theta d_1
 \le d_{0,G}+\theta d_{1,G}-\eta,
 \qquad
 d_{0,G}=V-q_G,\quad d_{1,G}=V+q_G-2c_G.
\]
For $P>0$, the variational identity
\eqref{eq:perception-variational} then yields
\[
 D_\mu(R,C,P)
 \le d_0+h_P(d_1)
 \le d_{0,G}+h_P(d_{1,G})-\eta
 =D_G(R,C,P)-\eta.
\]
For $P=0$, the same conclusion follows by taking $\theta=1$. Thus the
distortion comparison is strict for every $P\ge0$ whenever the
minimizing allocation has positive $r_i$ in every direction.


Finally, when $C=\infty$, the RDPC comparison reduces to the RDP
comparison in \Cref{thm:rdp-extremality}, including its strictness
under the positive-allocation condition. 
For a scalar source, $R>0$ gives a positive rate in the only direction in both Gaussian benchmark formulas, yielding the stated strict comparison.
\end{proof}

The construction also shows how common randomness changes the midpoint
geometry. At $C=0$, the Gaussian optimizer has $A=B$, so the auxiliary
can coincide with the midpoint. For $C>0$, the larger information
budget for the second endpoint allows a stronger observation along the
right half-bridge. The two half-bridges still meet at the same smoothed
marginal. The Gaussian formulas
\eqref{eq:gaussian-parallel}--\eqref{eq:gaussian-scalar-rdpc}
quantify the resulting distortion tradeoff; their special cases are
discussed in \Cref{app:gaussian-special-cases}.

\section{Conclusion}\label{sec:conclusion}

The self-coupling inequality gives a sharp Gaussian comparison for the
information cost of imposing self-correlation. Its proof constructs
two half-bridges with the same midpoint law and joins their conditional
endpoint distributions. The relative entropy of the shared midpoint yields a
gain that, through a Gaussian change of reference and convex duality,
leads to a matrix cross-covariance bound. For non-Gaussian laws, the
information comparison is strict whenever the target cross-covariance
is positive definite. In particular, every non-Gaussian scalar law
admits strictly greater self-correlation than the Gaussian benchmark
at every positive information budget.

The quadratic RDP comparison follows by scaling
one endpoint of the self-coupling. With finite common randomness, two
information budgets constrain the endpoints' dependence on a shared
auxiliary random variable. A variational formula for the squared Wasserstein-2 perception
penalty reduces the problem to comparing two weighted quadratic errors.
A marked midpoint bridge, with its rates balanced simultaneously using
the Poincar\'e--Miranda theorem, then establishes Gaussian extremality
in this setting as well. Both distortion comparisons are strict for
non-Gaussian sources when a Gaussian minimizing allocation assigns
positive rate to every covariance eigendirection.

The shared midpoint entropy also raises questions about quantitative
stability of these comparisons. More broadly, gluing several
conditional endpoint laws over a common intermediate distribution may
help extend the analysis to self-coupling problems with more than two
copies. In both directions, the key idea is to coordinate the endpoint
constraints through a shared intermediate law.

\medskip
\noindent\textbf{Disclosure on the use of generative AI.}
GPT 5.6 Pro-Sol assisted in refining the authors' initial RDPC proof
strategy by suggesting adjustable observation strengths for simultaneous
rate balancing via the Poincar\'e--Miranda theorem. It also assisted in
constructing the counterexamples in
\Cref{app:rdp-nonstrict-vector,app:fixed-rd-representation}.
The underlying framework and initial RDPC construction were developed
by the authors, who take full responsibility for the final arguments
and conclusions.

\appendices
\crefalias{section}{appendix}
\section{Attainment and duality for the self-coupling information cost}
\label{app:duality}

Let $\nu$ have mean zero and covariance $I_d$. Extend
\eqref{eq:Jmu} to all real symmetric matrices $S$, setting
$J_\nu(S)=+\infty$ whenever the constraint is infeasible. We show that
$J_\nu$ is proper, convex, and lower semicontinuous, that every finite
value is attained, and that its conjugate is the free energy in
\eqref{eq:free-energy}.

Convexity follows by mixing couplings. If $\pi_0$ and $\pi_1$ are
feasible for targets $S_0$ and $S_1$, respectively, their mixture is
feasible for the corresponding mixture of the targets. Relative
entropy $\mathrm{KL}(\pi\|\nu\otimes\nu)$ is convex in $\pi$
\cite{polyanskiy2025information}. The independent coupling has zero
mutual information and zero cross-covariance, so $J_\nu(0)=0$.
In particular, $J_\nu$ is proper.

We next establish the compactness needed to handle the unbounded
quadratic constraint. The couplings in $\Pi(\nu,\nu)$ form a tight
family. Moreover, for each $i,j$, the products $X_iY_j$ are uniformly
integrable over this family. Indeed, for $a>0$, Cauchy--Schwarz gives
\begin{align*}
 \E_\pi\!\left[|X_iY_j|
             \boldsymbol1_{\{|X_iY_j|>a\}}\right]
        \le
 \left(\E_\nu[X_i^2\boldsymbol1_{\{|X_i|>\sqrt a\}}]
              \E_\nu[Y_j^2]\right)^{1/2}
 +\left(\E_\nu[X_i^2]
              \E_\nu[Y_j^2\boldsymbol1_{\{|Y_j|>\sqrt a\}}]
       \right)^{1/2}.
\end{align*}
This bound is uniform over $\pi$ and tends to zero as $a\to\infty$.
Consequently, weak convergence within this family implies entrywise
convergence of the cross-covariance matrices, by truncating each
continuous function $x_i y_j$.

To prove lower semicontinuity, let $S_n\to S$ in $\mathbb S^d$ and
suppose that
\[
 L:=\liminf_{n\to\infty}J_\nu(S_n)<\infty.
\]
Passing to a subsequence, which we continue to denote by $(S_n)$, assume
that $J_\nu(S_n)\to L$. For each $n$, choose
$\pi_n\in\Pi(\nu,\nu)$ such that
\[
 \operatorname{Sym}\int xy^\top\,\pi_n(\dd x,\dd y)\succeq S_n
\]
and
\[
 \mathrm{KL}(\pi_n\|\nu\otimes\nu)
 \le J_\nu(S_n)+\frac1n.
\]
Since all $\pi_n$ have the same marginals, the sequence $(\pi_n)$ is
tight. After passing to a further subsequence, suppose that
$\pi_n\rightharpoonup\pi^\star$ for some $\pi^\star\in\Pi(\nu,\nu)$, where $\rightharpoonup$ denotes  weak convergence.
The uniform integrability established above gives
\[
 \int xy^\top\,\dd \pi_n\longrightarrow\int xy^\top\,\dd \pi^\star.
\]
Since the positive-semidefinite cone is closed and $S_n\to S$, it
follows that
\[
 \operatorname{Sym}\int xy^\top\,\dd \pi^\star\succeq S.
\]
Thus $\pi^\star$ is feasible for $J_\nu(S)$. Lower semicontinuity of
relative entropy now yields
\[
 \begin{aligned}
 J_\nu(S) \le \mathrm{KL}(\pi^\star\|\nu\otimes\nu)\le \liminf_{n\to\infty}\mathrm{KL}(\pi_n\|\nu\otimes\nu)\le \lim_{n\to\infty}\left(J_\nu(S_n)+\frac1n\right)
 =L.
 \end{aligned}
\]
Therefore $J_\nu$ is lower semicontinuous.

Taking $S_n=S$ in the same argument gives a feasible limit $\pi^\star$
with $\KL{\pi^\star}{\nu\otimes\nu}=J_\nu(S)$ whenever
$J_\nu(S)<\infty$. Thus every finite value is attained.

Equip the space of symmetric matrices with the trace inner product.
For a coupling $\pi$, write
$H_\pi=\Sym\E_\pi[XY^\top]$. By the definition of the conjugate,
\begin{align*}
 J_\nu^*(M)
 &=\sup_S\sup_{\pi:\,S\preceq H_\pi}
       \{\Tr(MS)-I_\pi(X;Y)\}\\
 &=\sup_\pi\left\{
       \sup_{S\preceq H_\pi}\Tr(MS)-I_\pi(X;Y)\right\}.
\end{align*}
If $M\succeq0$, the inner supremum is $\Tr(MH_\pi)$ and is attained
at $S=H_\pi$. If $v^\top Mv<0$, choosing
$S=H_\pi-tv v^\top$ makes the inner supremum tend to $+\infty$ as
$t\to\infty$. Therefore,
\[
 J_\nu^*(M)=
 \begin{cases}
 \Lambda_\nu(M),&M\succeq0,\\
 +\infty,&M\not\succeq0.
 \end{cases}
\]
The Fenchel--Moreau theorem now gives \eqref{eq:fenchel} on the entire
space of symmetric targets. Finally,
$\Sym\E_\pi[XY^\top]\preceq I_d$, since
$\E[(X-Y)(X-Y)^\top]\succeq0$. Thus, $0\le\Lambda_\nu(M)\le\Tr M$ for $M\succeq0$, so the free energy
used throughout the proof is finite.

\section{Gaussian smoothing for the free-energy comparison}
\label{app:smoothing}


We extend \Cref{prop:free-energy} from laws with finite relative
entropy to any centered law $\nu$ with covariance $I_d$. Fix
$M\succeq0$, let $S$ be given by \eqref{eq:S-from-M}, and set
$T=S^{1/2}$. We approximate $\nu$ by Gaussian-smoothed laws $\nu_n$
with finite relative entropy with respect to $\gamma$, apply the
established comparison to each $\nu_n$, and pass to the limit using
lower semicontinuity.

Choose $0<c_n\uparrow1$. For independent $X\sim\nu$ and
$G\sim\gamma$, define
\[
 X_n=c_n X+\sqrt{1-c_n^2}\,G,\qquad
 \nu_n=\Law(X_n)=\nu K_{c_nI_d}.
\]
Each $\nu_n$ has mean zero and covariance $I_d$, and
$W_2^2(\nu_n,\nu)\le\E\norm{X_n-X}^2=2d(1-c_n)\to0$.
The added Gaussian noise gives $\nu_n$ a smooth, positive density.
Convexity of relative entropy also gives
\[
 \KL{\nu_n}{\gamma}
 \le\int\KL{\mcN(c_n x,(1-c_n^2)I_d)}{\gamma}\,\nu(\dd x)
 =-\frac d2\log(1-c_n^2)<\infty.
\]
Here $\E\norm X^2=d$ cancels the quadratic terms in the Gaussian
relative-entropy formula. Thus the finite-entropy comparison applies
to each $\nu_n$.

By \Cref{app:duality}, the free energy is finite. Choose
$\pi_n\in\Pi(\nu_n,\nu_n)$ whose objective in \eqref{eq:free-energy}
is at least $\Lambda_{\nu_n}(M)-1/n$. Combining this choice with the
comparison for $\nu_n$ gives
\begin{equation}
   \E_{\pi_n} [X_n^\top MY_n] - I_{\pi_n}(X_n; Y_n) \ge \Lambda_{\nu_n}(M) - \frac{1}{n} \ge \Lambda_\gamma(M) - \frac{1}{n} + \KL{\nu_n K_T}{\gamma},
\end{equation}
or, equivalently,
\begin{equation}\label{eq:approximation-two-entropies}
   I_{\pi_n}(X_n; Y_n)+ \KL{\nu_n K_T}{\gamma} \le  \E_{\pi_n} [X_n^\top MY_n] - \Lambda_\gamma(M) + \frac{1}{n}.
\end{equation}
The marginals $\nu_n$ form a tight family, so the couplings $\pi_n$
are tight as well. Passing to a subsequence, suppose that
$\pi_n\rightharpoonup\pi\in\Pi(\nu,\nu)$. Since $\nu_n\to\nu$ in
$W_2$, their squared norms are uniformly integrable:
\[
 \lim_{L\to\infty}\sup_n
 \int_{\{\norm w>L\}}\norm w^2\,\nu_n(\dd w)=0.
\]
The Cauchy--Schwarz truncation estimate from \Cref{app:duality} then
shows that each product $X_{n,i}Y_{n,j}$ is uniformly integrable.
Consequently,
\begin{equation}\label{eq:approximation-reward-convergence}
 \E_{\pi_n}[X_n^\top MY_n]\longrightarrow\E_\pi[X^\top MY].
\end{equation}

We also have $\nu_n\otimes\nu_n\rightharpoonup\nu\otimes\nu$ and
$\nu_nK_T\rightharpoonup\nu K_T$, the latter by continuity of the
Gaussian observation kernel. Joint lower semicontinuity of relative
entropy, its nonnegativity, and
\eqref{eq:approximation-two-entropies}--\eqref{eq:approximation-reward-convergence}
therefore imply
\begin{align*}
 I_\pi(X;Y)+\KL{\nu K_T}{\gamma}
 &\le\liminf_n\left\{
       \KL{\pi_n}{\nu_n\otimes\nu_n}+\KL{\nu_nK_T}{\gamma}
                  \right\}\\
 &\le\E_\pi[X^\top MY]-\Lambda_\gamma(M).
\end{align*}
The right-hand side is finite, so both limiting divergences are finite.
Rearranging and using the feasible self-coupling $\pi$ gives
\[
 \Lambda_\nu(M)
 \ge\E_\pi[X^\top MY]-I_\pi(X;Y)
 \ge\Lambda_\gamma(M)+\KL{\nu K_{S^{1/2}}}{\gamma}.
\]
This proves \Cref{prop:free-energy}, including the midpoint entropy
gain.

\section{Equality for a non-Gaussian vector source}
\label{app:rdp-nonstrict-vector}

We now give the example mentioned in \Cref{rem:rdp-nonstrict-vector}.

Equality can occur when the Gaussian benchmark allocates all its rate
to a Gaussian component of a non-Gaussian source. Let $X=(Z,B)$, where
$Z\sim\mcN(0,4)$ and $B$ is equally likely to be $-1$ or $1$,
independently of $Z$. Then $\E B=0$, $\E B^2=1$, and $\mu=\Law(X)$ is
non-Gaussian with covariance $\diag(4,1)$. We show that
\[
 D_\mu(R_*,P_*)=D_G(R_*,P_*)=3,
 \qquad R_*=\tfrac12\log2,\quad P_*=5,
\]
where $G=\mcN(0,\diag(4,1))$.

\begin{enumerate}
\item We first evaluate the Gaussian benchmark. For every feasible rate
allocation, $0\le\rho(r_i)\le1$ implies
\[
 4(1-\rho(r_1))^2+(1-\rho(r_2))^2\le5=P_*.
\]
Hence the perception penalty in \eqref{eq:gaussian-rdp} vanishes, and
\[
 D_G(R_*,P_*)
 =\min_{\substack{r_1,r_2\ge0\\r_1+r_2\le R_*}}
       \bigl(4e^{-2r_1}+e^{-2r_2}\bigr).
\]
Both terms decrease with their respective rates, so the rate constraint
is tight at a minimizer. Set $s=r_2$ and use $e^{-2R_*}=1/2$ to write
\[
 f(s)=4e^{-2(R_*-s)}+e^{-2s}=2e^{2s}+e^{-2s},
 \qquad 0\le s\le R_*.
\]
Since $f'(s)=4e^{2s}-2e^{-2s}\ge2>0$, the minimum occurs at $s=0$.
Hence the minimizing allocation is $(r_1,r_2)=(R_*,0)$, and
$D_G(R_*,P_*)=f(0)=3$.

\item An explicit reconstruction for $\mu$ attains distortion $3$. Let
$N\sim\mcN(0,1)$ be independent of $(Z,B)$, and set
$T=Z/2+N$ and $\widehat X=(T,0)$. Since
$\Var(T)=4/4+1=2$, we have $T\sim\mcN(0,2)$. Also, $B$ is independent
of $(Z,T)$, so
\begin{align*}
 I(X;\widehat X)
 &=I(Z,B;T)=I(Z;T)\\
 &=h(T)-h(T\mid Z)=h(T)-h(N)\\
 &=\tfrac12\log(2\pi e\cdot2)-\tfrac12\log(2\pi e)
 =\tfrac12\log2=R_*.
\end{align*}
The distortion is
\begin{align*}
 \E\norm{X-\widehat X}^2
 &=\E(Z/2-N)^2+\E B^2\\
 &=\tfrac14\E Z^2+\E N^2+\E B^2
 =1+1+1=3,
\end{align*}
since the cross term vanishes by independence and zero means.

To check the perception constraint, note that the second-coordinate
cost is $1$ in every coupling of $\mu$ and $\Law(\widehat X)$. The
first coordinates have variances $4$ and $2$, so Cauchy--Schwarz bounds
their cross-moment by $\sqrt{4\cdot2}$. Thus the total transport cost
is at least
$4+2-2\sqrt{4\cdot2}+1=7-4\sqrt2$.
This lower bound is attained by coupling $(Z,B)$ with
$(Z/\sqrt2,0)$, whose second marginal is $\Law(\widehat X)$. Hence
\[
 W_2^2(\mu,\Law(\widehat X))=7-4\sqrt2<5=P_*.
\]
The reconstruction is feasible, so $D_\mu(R_*,P_*)\le3$. This upper
bound also follows from \Cref{thm:rdp-extremality}.

\item To prove equality, it remains to rule out distortion below $3$. Let
$V=(V_1,V_2)$ be any reconstruction with
$I(Z,B;V)\le R_*$ and finite distortion, and set $t=I(B;V)$.
Since $Z$ and $B$ are independent, the chain rule gives
\begin{align*}
 I(Z,B;V)
 &=I(Z;V)+I(B;V\mid Z)\\
 &=I(Z;V)+I(B;V)+I(Z;B\mid V)\\
 &\ge I(Z;V)+t.
\end{align*}
Thus $0\le t\le R_*$ and $I(Z;V)\le R_*-t$. In other words, the
information used to represent the binary coordinate reduces the
information available for the Gaussian coordinate.

Let $d_Z=\E(Z-V_1)^2$. Translation by the $V$-measurable variable
$V_1$, followed by removal of conditioning and the maximum-entropy
bound, gives
\[
 h(Z\mid V)=h(Z-V_1\mid V)
 \le h(Z-V_1)\le\tfrac12\log(2\pi e\,d_Z).
\]
Since $h(Z)=\tfrac12\log(2\pi e\cdot4)$, it follows that
\[
 I(Z;V)\ge\tfrac12\log\frac4{d_Z},
 \qquad
 d_Z\ge4e^{-2I(Z;V)}\ge4e^{-2(R_*-t)}=2e^{2t}.
\]

For the binary coordinate, define
$p(V)=\Pr(B=1\mid V)$ and $b(V)=\E[B\mid V]=2p(V)-1$. The
conditional relative-entropy formula for mutual information gives
\[
 t=\E\bigl[p(V)\log(2p(V))
       +(1-p(V))\log(2(1-p(V)))\bigr].
\]
For $0\le p\le1$, define
$\phi(p)=p\log(2p)+(1-p)\log(2(1-p))$, with $0\log0=0$.
Since $\phi(1/2)=\phi'(1/2)=0$ and
$\phi''(p)=1/[p(1-p)]\ge4$ on $(0,1)$, we have
$\phi(p)\ge2(p-1/2)^2$; the inequality extends to the endpoints by
continuity. Therefore,
\[
 t\ge2\E\bigl(p(V)-\tfrac12\bigr)^2
   =\tfrac12\E b(V)^2.
\]
Conditional-mean orthogonality then gives
\begin{align*}
 \E(B-V_2)^2
 &=\E(B-b(V))^2+\E(b(V)-V_2)^2\\
 &\ge\E B^2-\E b(V)^2
 =1-\E b(V)^2\ge1-2t.
\end{align*}
Combining the bounds for the two coordinates and using
$e^{2t}\ge1+2t$, 
we obtain
\[
 \E\norm{X-V}^2
 \ge2e^{2t}+1-2t
 \ge2(1+2t)+1-2t=3+2t\ge3.
\]
This lower bound holds without imposing a perception constraint, so it
applies to every feasible RDP reconstruction. Together with the upper
bound, it proves
$D_\mu(R_*,P_*)=D_G(R_*,P_*)=3$.
\end{enumerate}
The Gaussian optimizer assigns zero rate to the direction carrying the
non-Gaussian component. Thus non-Gaussianity alone does not guarantee
strictness in every dimension for every choice of $(R,P)$.

\section{Suboptimality of a fixed RD representation at perfect perception}
\label{app:fixed-rd-representation}
\label{app:posterior-interpolation}

We give the two-dimensional example discussed in
\Cref{rem:conventional-rd-interpolation}. It shows that the rate
allocation of an ordinary RD optimizer need not remain optimal when
perfect perception is imposed. Let
\[
 G\sim\mcN\bigl(0,\diag(4,1)\bigr),\qquad R=\frac12\log2.
\]
Gaussian reverse water filling allocates rates $(R,0)$ to the two
coordinates. The resulting reconstruction satisfies
\[
 \widehat{X}_{\rm RD}\sim\mcN\bigl(0,\diag(2,0)\bigr),\qquad
 \varepsilon_{\rm RD}=3.
\]
The squared Wasserstein-2 distance between the source and reconstruction
laws is
\[
 w_{\rm RD}^2=(2-\sqrt2)^2+1.
\]
Any decoder based on this fixed representation and satisfying perfect
perception induces a Markov chain
$G-\widehat{X}_{\rm RD}-Y$, with $Y\sim\Law(G)$. Conditional-mean
orthogonality and the definition of $W_2$ give
\begin{equation}\label{eq:fixed-rd-exact-perception}
 \E\norm{G-Y}^2
 =\varepsilon_{\rm RD}+\E\norm{\widehat{X}_{\rm RD}-Y}^2
 \ge\varepsilon_{\rm RD}+w_{\rm RD}^2
 =10-4\sqrt2.
\end{equation}
This bound is attained by choosing the conditional law of $Y$ given
$\widehat{X}_{\rm RD}$ from an optimal Wasserstein-2 coupling.

Now consider the rate allocation
\[
 r_1=\frac12\log\frac95,\qquad
 r_2=\frac12\log\frac{10}{9}.
\]
It satisfies $r_1+r_2=R$ and gives
$\rho(r_1)=2/3$ and $\rho(r_2)=1/\sqrt{10}$. Substituting this
allocation into \eqref{eq:gaussian-rdp} at $P=0$ yields
\begin{equation}\label{eq:fixed-rd-vector-obstruction}
 D_G(R,0)
 \le 2\left[4\left(1-\frac23\right)
                 +1-\frac1{\sqrt{10}}\right]
 =\frac{14}{3}-\frac2{\sqrt{10}}
 <10-4\sqrt2.
\end{equation}
Thus, even for this Gaussian source, the fixed ordinary RD
representation is suboptimal at perfect perception.

\section{Weighted variational representation of the RDPC function}
\label{app:rdpc-variational}

We prove \Cref{prop:rdpc-variational} by showing that the error region
is compact and convex, then applying a minimax theorem to the perception
penalty. We also prove continuity of the weighted value and attainment
of the dual optimum. Throughout, fix finite $R,C\ge0$.

\subsection{Compactness and convexity of the error region}

Recall that $\mathcal D_{\rm att}$ is the set of error pairs
$(d_0,d_1)$ in \eqref{eq:rdpc-two-errors}  generated by feasible chains
in $\cA_\mu(R,C)$, and let $\mathcal D$ be its Euclidean closure.
Write $m=\E X$ and $V=\Tr\Sigma$. For any such chain, the conditional
mean $M_U=\E[X\mid U]$ satisfies
$\E\norm{M_U-m}^2\le V$. By conditional-mean orthogonality and the
squared triangle inequality,
\[
 \begin{split}
 d_0&=V-\E\norm{M_U-m}^2\in[0,V],\\
 d_1&\le2\E\norm{Y-m}^2+2\E\norm{M_U-m}^2\le4V.
 \end{split}
\]
The region is nonempty, since independent endpoints with a constant
auxiliary form a feasible chain. Hence $\mathcal D$ is a nonempty
compact subset of $[0,V]\times[0,4V]$. If $V=0$, the source is
deterministic and all error values are zero, so suppose henceforth
that $V>0$.

To prove convexity, take two feasible chains
$X_j-U_j-Y_j$, $j=0,1$, and define
\[
 m_j(u)=\E[X_j\mid U_j=u],\qquad
 d_{0,j}=\E\norm{X_j-m_j(U_j)}^2,\qquad
 d_{1,j}=\E\norm{Y_j-m_j(U_j)}^2.
\]
Fix $0<t<1$ and set $p_0=1-t$, $p_1=t$. Draw a label $J$ with
$\Pr(J=j)=p_j$. Conditional on $J=j$, let $(X,U,Y)$ have the joint
law of $(X_j,U_j,Y_j)$, and define the new auxiliary
$\widetilde U=(J,U)$. Given $\widetilde U=(j,u)$, the endpoints have
the conditional laws of $X_j$ and $Y_j$ given $U_j=u$. They are
conditionally independent, so $X-\widetilde U-Y$ is a Markov chain.

For each value of $J$, both endpoint marginals are $\mu$. Thus $J$ is
independent of both $X$ and $Y$, and the mutual-information chain rule
gives
\begin{align*}
 I(X;\widetilde U)
 =I(X;J)+I(X;U_J\mid J)
   =\sum_{j=0}^1p_j I(X_j;U_j)\le R,\;\;I(Y;\widetilde U)
   =\sum_{j=0}^1p_j I(Y_j;U_j)\le R+C.
\end{align*}
The mixed chain is therefore feasible. Because the auxiliary retains
the label,
\[
 \widetilde M:=\E[X\mid\widetilde U]=m_J(U_J).
\]
Thus, conditional on $J=j$, $\widetilde M=m_j(U_j)$. Averaging the
squared errors over $J$ gives
\begin{align*}
 \widetilde d_0
 =\sum_{j=0}^1p_j\E\bigl[\norm{X-\widetilde M}^2\mid J=j\bigr]
   =\sum_{j=0}^1p_j d_{0,j},\;\; \widetilde d_1
 =\sum_{j=0}^1p_j\E\bigl[\norm{Y-\widetilde M}^2\mid J=j\bigr]
   =\sum_{j=0}^1p_j d_{1,j}.
\end{align*}
Hence every convex combination of two attainable error pairs is
attainable. Therefore $\mathcal D_{\rm att}$, and consequently its
closure $\mathcal D$, is convex.

Taking the closure does not change the infimum of a continuous
objective, and compactness ensures that the infimum is attained.
Recalling $h_P(v)=(\sqrt v-\sqrt P)_+^2$, we obtain
\[
 D_\mu(R,C,P)
 =\min_{(d_0,d_1)\in\mathcal D}
 \{d_0+h_P(d_1)\},\;\;\cE_{\mu,\theta}(R,C)=\min_{(d_0,d_1)\in\mathcal D}
 (d_0+\theta d_1),\qquad 0\le\theta\le1.
\]

\subsection{Minimax representation of the RDPC function}

Assume first that $P>0$. For fixed $d_1$, the maximizing weight in
\eqref{eq:perception-variational} is $0$ if $d_1\le P$ and
$1-\sqrt{P/d_1}$ if $d_1>P$. Since every $(d_0,d_1)\in\mathcal D$
satisfies $d_1\le4V$, these maximizing weights all lie in the compact
interval
\[
 0\le\theta\le\bar\theta,
 \qquad \bar\theta=\max\{0,1-\sqrt{P/(4V)}\}<1.
\]
Consequently, with
\[
 \Phi(d_0,d_1,\theta)
 :=d_0+\theta d_1-\frac{\theta P}{1-\theta},
\]
the scalar variational identity gives
$d_0+h_P(d_1)=\max_{0\le\theta\le\bar\theta}\Phi(d_0,d_1,\theta)$
for every $(d_0,d_1)\in\mathcal D$.

The sets $\mathcal D$ and $[0,\bar\theta]$ are compact and convex.
Moreover, $\Phi$ is continuous, affine in $(d_0,d_1)$, and concave in
$\theta$. Sion's minimax theorem \cite{Sion1958} therefore gives
\begin{align*}
 D_\mu(R,C,P)
 &=\min_{(d_0,d_1)\in\mathcal D}
             \max_{0\le\theta\le\bar\theta}\Phi(d_0,d_1,\theta)\\
 &=\max_{0\le\theta\le\bar\theta}
             \min_{(d_0,d_1)\in\mathcal D}\Phi(d_0,d_1,\theta)\\
 &=\max_{0\le\theta\le\bar\theta}
       \left\{\cE_{\mu,\theta}(R,C)-\frac{\theta P}{1-\theta}\right\}.
\end{align*}
It remains to extend the maximization to $0\le\theta<1$. For any such
$\theta$, the scalar variational identity gives
$\Phi(d_0,d_1,\theta)\le d_0+h_P(d_1)$ for every error pair. Taking
the minimum over $\mathcal D$ yields
\[
 \cE_{\mu,\theta}(R,C)-\frac{\theta P}{1-\theta}
 \le D_\mu(R,C,P).
\]
Since the maximum over the smaller interval already equals
$D_\mu(R,C,P)$, extending the interval does not change the value. This
proves \eqref{eq:rdpc-energy-duality}. When $P=0$, no minimax interchange is needed: $h_0(d_1)=d_1$, so
$D_\mu(R,C,0)=\cE_{\mu,1}(R,C)$, as stated in
\eqref{eq:rdpc-exact-energy}.

\subsection{Continuity of the weighted value}

Fix $\theta,\theta'\in[0,1]$, and let
$(d_0',d_1')\in\mathcal D$ minimize the objective at weight $\theta'$.
Using this pair as a competitor at weight $\theta$ gives
\begin{align*}
 \cE_{\mu,\theta}(R,C)
 &\le d_0'+\theta d_1'\\
 &=\cE_{\mu,\theta'}(R,C)+(\theta-\theta')d_1'\\
 &\le\cE_{\mu,\theta'}(R,C)+4V|\theta-\theta'|.
\end{align*}
Interchanging $\theta$ and $\theta'$ yields the reverse bound, so
\[
 \bigl|\cE_{\mu,\theta}(R,C)-\cE_{\mu,\theta'}(R,C)\bigr|
 \le4V|\theta-\theta'|.
\]
Thus $\cE_{\mu,\theta}(R,C)$ is Lipschitz continuous in $\theta$ on
$[0,1]$.
For $P>0$, dual attainment follows from the compact maximization above:
a maximizing weight can be chosen in $[0,\bar\theta]$, where
$\bar\theta<1$.

\section{Gaussian optimizers and ordered observation strengths}
\label{app:gaussian-structure}

We prove the Gaussian representation and ordering asserted in
\Cref{lem:gaussian-normal-form}. Throughout, we work in the covariance
coordinates of \Cref{subsec:gaussian-moment-targets}, where
$\Sigma=\diag(\sigma_1^2,\ldots,\sigma_d^2)\succ0$ and
$V=\Tr\Sigma$.

\subsection{Evaluation of the Gaussian weighted problem}

Consider any feasible chain $G-Z-\widetilde G$ with
$G,\widetilde G\sim\mcN(0,\Sigma)$. Define
$M_i=\E[G_i\mid Z]$, $N_i=\E[\widetilde G_i\mid Z]$, and
\[
 r_i=I(G_i;Z),\qquad s_i=I(\widetilde G_i;Z).
\]
The coordinates of $G$ are independent, so $h(G)=\sum_i h(G_i)$.
Conditional entropy subadditivity  gives
\[
 I(G;Z)=\sum_i h(G_i)-h(G\mid Z)
 \ge\sum_i\bigl[h(G_i)-h(G_i\mid Z)\bigr]
 =\sum_i r_i.
\]
Applying the same argument to $\widetilde G$ yields
\[
 \sum_i r_i\le I(G;Z)\le R,\qquad
 \sum_i s_i\le I(\widetilde G;Z)\le R+C.
\]
The scalar Gaussian quadratic RD bound
\cite{CoverThomas2006,Berger1971} implies
\begin{equation}\label{eq:normal-form-posterior-bound}
 \E M_i^2\le\sigma_i^2\rho(r_i)^2,\qquad
 \E N_i^2\le\sigma_i^2\rho(s_i)^2.
\end{equation}
Since $G$ and $\widetilde G$ are conditionally independent given $Z$,
Cauchy--Schwarz then gives
\begin{equation}\label{eq:normal-form-cross-bound}
 \E[G_i\widetilde G_i]=\E[M_iN_i]
 \le\sigma_i^2\rho(r_i)\rho(s_i).
\end{equation}
It follows that
\[
 \begin{aligned}
 d_0+\theta d_1
 &=(1+\theta)V-\sum_i\bigl[(1-\theta)\E M_i^2
                              +2\theta\E[M_iN_i]\bigr]\\
 &\ge(1+\theta)V-\sum_i\sigma_i^2
       \bigl[(1-\theta)\rho(r_i)^2+2\theta\rho(r_i)\rho(s_i)\bigr]
 \end{aligned}
\]
for $0\le\theta\le1$.
Optimizing over the feasible coordinate rates gives the lower bound
in \eqref{eq:gaussian-weighted-energy}.

To attain this bound, set $a_i=\rho(r_i)$ and $b_i=\rho(s_i)$. Let
the standard Gaussian variables $Z_i,E_i,\widetilde E_i$ be mutually
independent, and define
\begin{equation}\label{eq:normal-form-parallel-chain}
 \begin{split}
 G_i&=\sigma_i\bigl(a_iZ_i+\sqrt{1-a_i^2}\,E_i\bigr),\\
 \widetilde G_i
 &=\sigma_i\bigl(b_iZ_i+\sqrt{1-b_i^2}\,\widetilde E_i\bigr).
 \end{split}
\end{equation}
Given $Z=(Z_1,\ldots,Z_d)$, the endpoints are independent,
with conditional means $\sigma_i a_iZ_i$ and $\sigma_i b_iZ_i$.
The Gaussian channel formula gives
$I(G;Z)=\sum_i g(a_i)$ and
$I(\widetilde G;Z)=\sum_i g(b_i)$. Both moment bounds above hold with
equality, proving \eqref{eq:gaussian-weighted-energy}.


\subsection{Ordered strengths for the weighted problem}

Write the objective in terms of the rates:
\[
 \begin{aligned}
 F_\theta(r,s)&=\sum_i f_i(r_i,s_i),\\
 f_i(r,s)&=\sigma_i^2\bigl[(1-\theta)(1-e^{-2r})
                         +2\theta\rho(r)\rho(s)\bigr].
 \end{aligned}
\]
As before, set $a_i=\rho(r_i)$ and $b_i=\rho(s_i)$. The rate
simplices are compact, so a maximizer exists. If $R=0$, take
$A=B=0$. If $\theta=0$, the objective does not depend on $B$; choosing
$B=A$ for an optimal $A$ gives a feasible ordered pair.

It remains to consider $R>0$ and $\theta>0$. The maximum is positive,
so $a_i>0$ for some $i$. Increasing $s_i$ at this coordinate strictly
increases the objective, so the second rate budget must be fully used.
In particular, $b_j>0$ for some $j$. Increasing $r_j$ would then
strictly increase the objective as well, so the first budget is also
fully used:
\begin{equation}\label{eq:normal-form-full-budgets}
 \sum_i r_i=R,\qquad \sum_i s_i=R+C.
\end{equation}

We next show that the two allocations have the same active set.
Suppose first that $r_i=0<s_i$. Since $R>0$, some coordinate $j$ has
$r_j>0$. Keep $s$ fixed and replace $(r_i,r_j)$ by
$(\varepsilon,r_j-\varepsilon)$, where $0<\varepsilon<r_j/2$. This
preserves the first rate budget. The gain at coordinate $i$ is
\[
 \begin{aligned}
 f_i(\varepsilon,s_i)-f_i(0,s_i)
 &=\sigma_i^2\bigl[(1-\theta)(1-e^{-2\varepsilon})
                  +2\theta\rho(s_i)\rho(\varepsilon)\bigr]\\
 &=2\sqrt2\,\theta\sigma_i^2\rho(s_i)\sqrt\varepsilon
   +O(\varepsilon),
 \end{aligned}
\]
using $1-e^{-2\varepsilon}=2\varepsilon+O(\varepsilon^2)$ and
$\rho(\varepsilon)=\sqrt{2\varepsilon}+O(\varepsilon^{3/2})$.
The coefficient of $\sqrt\varepsilon$ is positive because
$\theta>0$ and $s_i>0$. At coordinate $j$, the function
$f_j(\cdot,s_j)$ is smooth near $r_j>0$, so
\[
 f_j(r_j-\varepsilon,s_j)-f_j(r_j,s_j) =-\partial_r f_j(r_j,s_j)\varepsilon+O(\varepsilon^2).
\]
The total change is therefore $c\sqrt\varepsilon+O(\varepsilon)$
for some $c>0$, and is positive for sufficiently small $\varepsilon$.
This contradicts optimality.

Now suppose that $s_i=0<r_i$. Transfer a small amount of the second
rate from some coordinate $j$ with $s_j>0$ to coordinate $i$, keeping
$r$ fixed. The gain at $i$ is
$2\theta\sigma_i^2\rho(r_i)\rho(\varepsilon)$, whose leading term is
a positive multiple of $\sqrt\varepsilon$. Since $s_j>0$, the loss at
$j$ is $O(\varepsilon)$. This again contradicts optimality. Hence
\[
 \mathcal I:=\{i:r_i>0\}=\{i:s_i>0\}.
\]
Both rates vanish outside $\mathcal I$.

On the active set, small rate transfers are possible in either
direction. At a maximum, the marginal gain from the first rate is
therefore the same at every active coordinate, as is the marginal gain
from the second rate. Equivalently, the two equality constraints in
\eqref{eq:normal-form-full-budgets} have multipliers
$\ell_r,\ell_s$, with Lagrangian
\[
 \mathcal L(r,s)=F_\theta(r,s)
 -\ell_r\left(\sum_i r_i-R\right)
 -\ell_s\left(\sum_i s_i-(R+C)\right).
\]
Stationarity on $\mathcal I$ gives
$\partial_{r_i}F_\theta=\ell_r$ and
$\partial_{s_i}F_\theta=\ell_s$. Using
\[
 \rho'(r)=\frac{1-\rho(r)^2}{\rho(r)},\qquad r>0,
\]
we obtain
\[
 \begin{aligned}
 \partial_{r_i}F_\theta
 &=2\sigma_i^2(1-a_i^2)
       \left(1-\theta+\theta\frac{b_i}{a_i}\right),\\
 \partial_{s_i}F_\theta
 &=2\theta\sigma_i^2a_i\frac{1-b_i^2}{b_i}.
 \end{aligned}
\]
These derivatives are strictly positive on $\mathcal I$. Define
$\lambda=\ell_r/2$ and $\eta=\ell_s/(2\theta)$. Then
\begin{align}
 \lambda&=\sigma_i^2(1-a_i^2)
              \left(1-\theta+\theta\frac{b_i}{a_i}\right),
 \label{eq:normal-form-kkt-a}\\
 \eta&=\sigma_i^2a_i\frac{1-b_i^2}{b_i},
 \label{eq:normal-form-kkt-b}
\end{align}
where $\lambda,\eta>0$ are common to all $i\in\mathcal I$.

Subtracting these equations gives the ordering identity
\begin{equation}\label{eq:normal-form-ordering-identity}
 \lambda-\eta
 =\sigma_i^2(b_i-a_i)
   \left[\frac{\theta(1-a_i^2)}{a_i}+\frac1{b_i}+a_i\right].
\end{equation}
The bracket is strictly positive, so $b_i-a_i$ has the same sign as
$\lambda-\eta$ for every active coordinate. If this sign were negative,
strict monotonicity of $\rho$ would imply $s_i<r_i$ throughout
$\mathcal I$, contradicting
$\sum_i s_i=R+C\ge R=\sum_i r_i$. Thus $A\preceq B$. If $C=0$,
equality of the total rates forces $A=B$; if $C>0$, then
$a_i<b_i$ for every active coordinate. Together with the cases
$R=0$ and $\theta=0$ considered above, this proves
\Cref{lem:gaussian-normal-form}.

\section{Continuity of posterior energy under Gaussian observations}
\label{app:posterior-continuity}

The rate balances in \eqref{eq:rate-stop} depend on the observation
strengths through a conditional expectation. We first prove continuity
for bounded scalar targets using the Gaussian observation densities.
We then truncate the target while keeping the observation fixed, and
use the $L^2$ contraction property of conditional expectation to bound
the truncation error uniformly over all observation matrices.

\begin{lemma}[Posterior-energy continuity]
\label{lem:posterior-continuity}
Let $W\in\R^d$ have finite second moment, and let $D_n,D$ be diagonal
matrices such that
\[
 D_n\longrightarrow D,\qquad 0\preceq D_n,D\preceq rI_d
 \quad\text{for some }0\le r<1.
\]
Let $G\sim\mcN(0,I_d)$ be independent of $W$, and define
\[
 Z_n=D_nW+(I_d-D_n^2)^{1/2}G,\qquad
 Z=DW+(I_d-D^2)^{1/2}G.
\]
Then the second-moment matrices of the posterior means converge
entrywise:
\begin{equation}\label{eq:posterior-continuity-result}
 \E\bigl[\E[W\mid Z_n]\E[W\mid Z_n]^\top\bigr]
 \longrightarrow
 \E\bigl[\E[W\mid Z]\E[W\mid Z]^\top\bigr].
\end{equation}
\end{lemma}

\begin{proof}
First consider a bounded measurable scalar target $\psi(W)$ with
$|\psi|\le M$, estimated from the observations. Write
$Z_D=DW+(I_d-D^2)^{1/2}G$. The conditional density of $Z_D$ given
$W=w$ is
\[
 k_D(z\mid w)
 =\frac{\exp\!\left[-\frac12(z-Dw)^\top(I_d-D^2)^{-1}(z-Dw)\right]}
        {(2\pi)^{d/2}\det(I_d-D^2)^{1/2}}.
\]
Since $I_d-D^2\succeq(1-r^2)I_d$,
\[
 0<k_D(z\mid w)\le [2\pi(1-r^2)]^{-d/2},
\]
uniformly in $z,w$ and in all the matrices under consideration. For
each fixed $z,w$, continuity of the matrix inverse and determinant
gives $k_{D_n}(z\mid w)\to k_D(z\mid w)$, including when some
diagonal entries of $D$ are zero.

Define
\[
 p_D(z)=\E k_D(z\mid W),\qquad
 N_D(z)=\E[\psi(W)k_D(z\mid W)],\qquad
 m_D(z)=\frac{N_D(z)}{p_D(z)}.
\]
The function $p_D$ is the strictly positive density of $Z_D$, and
$|N_D|\le Mp_D$, so $|m_D|\le M$. To verify that $m_D(Z_D)$ is a
version of $\E[\psi(W)\mid Z_D]$, let $\varphi$ be any bounded
measurable function. Fubini's theorem gives
\[
 \E[\psi(W)\varphi(Z_D)]
 =\int\varphi(z)N_D(z)\,\dd z
 =\E[m_D(Z_D)\varphi(Z_D)].
\]
No density for $W$ is needed. The uniform likelihood bound and
$|\psi|\le M$ allow us to apply dominated convergence under the law of
$W$, yielding $p_{D_n}(z)\to p_D(z)$ and
$N_{D_n}(z)\to N_D(z)$ for every $z$. Since $p_D(z)>0$, it follows
that $m_{D_n}(z)\to m_D(z)$.

The densities also converge in $L^1$. Indeed,
\[
 \int|p_{D_n}-p_D|\,\dd z
 =2-2\int\min\{p_{D_n},p_D\}\,\dd z\longrightarrow0,
\]
because the integrand in the last integral converges pointwise to
$p_D$ and is dominated by $p_D$. Since $|m_{D_n}|,|m_D|\le M$,
\begin{align*}
 &\left|\int m_{D_n}(z)^2p_{D_n}(z)\,\dd z
              -\int m_D(z)^2p_D(z)\,\dd z\right|\\
 &\qquad\le M^2\int|p_{D_n}-p_D|\,\dd z
       +\int|m_{D_n}^2-m_D^2|p_D\,\dd z
 \longrightarrow0.
\end{align*}
The first term tends to zero by $L^1$ convergence. The second tends
to zero by dominated convergence, since its integrand is bounded by
$2M^2p_D$. This proves continuity of the posterior energy for every
bounded target.

Now fix $v\in\R^d$ and let $\xi=v^\top W$. Truncate this scalar
target as $\xi_K=\max\{-K,\min\{\xi,K\}\}$, leaving the observations
$Z_n,Z$ unchanged. Factoring the difference of squares, applying
Cauchy--Schwarz, and using the $L^2$ contraction property of
conditional expectation gives
\begin{align*}
 \left|\E\bigl[\E[\xi\mid Z_n]^2\bigr]
       -\E\bigl[\E[\xi_K\mid Z_n]^2\bigr]\right|
 &=\left|\E\bigl[\E[\xi-\xi_K\mid Z_n]
                    \E[\xi+\xi_K\mid Z_n]\bigr]\right|\\
 &\le\|\xi-\xi_K\|_2\|\xi+\xi_K\|_2\\
 &\le2\|\xi\|_2\|\xi-\xi_K\|_2.
\end{align*}
For the last inequality, we used $|\xi_K|\le|\xi|$, which implies
$\|\xi_K\|_2\le\|\xi\|_2$. This bound is uniform in $n$ and also holds
with $Z$ in place of $Z_n$. For each fixed $K$, the bounded-target
result applies to $\xi_K$. The triangle inequality therefore  yields
\[
 \limsup_{n\to\infty}
 \left|\E\bigl[\E[\xi\mid Z_n]^2\bigr]
       -\E\bigl[\E[\xi\mid Z]^2\bigr]\right|
 \le4\|\xi\|_2\|\xi-\xi_K\|_2.
\]
Since $\xi$ is square-integrable,
$\|\xi-\xi_K\|_2^2\le\E[\xi^2\boldsymbol1_{\{|\xi|>K\}}]\to0$.
Letting $K\to\infty$ proves the scalar assertion.

Taking $\xi=W_i$ gives convergence of the diagonal entries in
\eqref{eq:posterior-continuity-result}. For $i\ne j$, use
\[
 \begin{split}
 2\E\bigl[\E[W_i\mid Z_n]\E[W_j\mid Z_n]\bigr]
 &=\E\bigl[\E[W_i+W_j\mid Z_n]^2\bigr]\\
 &\quad-\E\bigl[\E[W_i\mid Z_n]^2\bigr]
       -\E\bigl[\E[W_j\mid Z_n]^2\bigr].
 \end{split}
\]
Each term on the right converges by the scalar result. This proves
convergence of the off-diagonal entries as well, completing the proof.
\end{proof}

In the marked construction, $t_i(\alpha)=\sqrt{b_i\alpha_i}$ is
continuous on the entire parameter box, including its boundary. Also,
$0\preceq T(\alpha)\preceq B\preceq rI_d$, where
$r=\max_i b_i<1$. Applying the lemma to any sequence
$\alpha^{(n)}\to\alpha$ shows that $Q_\nu(T(\alpha))$ depends
continuously on $\alpha$. Hence
$e_i(\alpha)=Q_\nu(T(\alpha))_{ii}-b_i\alpha_i$ is continuous as well.
The  factor
\[
 k_i(\alpha)=\frac{b_i-\alpha_i}{1-b_i\alpha_i}
\]
is continuous because $0\le\alpha_i\le a_i\le b_i<1$ implies
$1-b_i\alpha_i\ge1-b_i^2>0$. Finally, $g$ and $g'$ are continuous on
$[0,a_i]$. Therefore each component of the rate-balancing map in
\eqref{eq:rate-stop} is continuous on the closed box, as required by
the Poincar\'e--Miranda theorem.

\section{Special cases of the Gaussian RDPC formula}
\label{app:gaussian-special-cases}

We record three special cases of the Gaussian formula
\eqref{eq:gaussian-parallel}: no common randomness, perfect
perception, and unlimited common randomness. Throughout, let
$V=\Tr\Sigma$, and denote the positive eigenvalues of $\Sigma$ by
$\sigma_i^2$.

\subsection{No common randomness}

When $C=0$, the ordered optimizer in
\Cref{lem:gaussian-normal-form} has $r_i=s_i$ for every $i$. The two
squared errors then both equal $\sum_i\sigma_i^2e^{-2r_i}$. Define
the ordinary Gaussian quadratic RD function by
\begin{equation}\label{eq:ordinary-gaussian-rd}
 d_G(R)=\min_{\substack{r_i\ge0\\\sum_i r_i\le R}}
                  \sum_i\sigma_i^2e^{-2r_i}.
\end{equation}
For each $\theta$, the minimum weighted energy is therefore
$(1+\theta)d_G(R)$. It follows that
\begin{equation}\label{eq:gaussian-no-common-randomness}
 D_G(R,0,P)
 =d_G(R)+\pos{\sqrt{d_G(R)}-\sqrt P}^{\,2}.
\end{equation}
For $R>0$, the usual reverse water-filling formula gives
\[
 d_G(R)=\sum_i\min\{\sigma_i^2,\xi\},\qquad
 R=\frac12\sum_i\left[\log\frac{\sigma_i^2}{\xi}\right]_+,
\]
where $\xi>0$ is chosen to satisfy the rate equation. At $R=0$,
$d_G(0)=V$. In particular, at perfect perception,
$D_G(R,0,0)=2d_G(R)$.

\subsection{Perfect perception}

At perfect perception, $P=0$, the two squared-error terms in
\eqref{eq:gaussian-parallel} combine to give
\begin{equation}\label{eq:gaussian-perfect-perception}
 D_G(R,C,0)
 =2V-2\max_{\substack{r_i,s_i\ge0\\
                      \sum_i r_i\le R,\ \sum_i s_i\le R+C}}
                    \sum_i\sigma_i^2\rho(r_i)\rho(s_i).
\end{equation}
Both endpoints have the source law, so the distortion is determined
entirely by their achievable cross-correlation.

\subsection{Unlimited common randomness}

By \eqref{eq:rdp-rdpc-identification}, unlimited common randomness
reduces the problem to the RDP function evaluated in
\eqref{eq:gaussian-rdp}. The finite-budget Gaussian formula converges
to this limit:
\begin{equation}\label{eq:gaussian-unlimited-benchmark}
 \lim_{C\to\infty}D_G(R,C,P)=D_G(R,\infty,P)=D_G(R,P).
\end{equation}
Indeed, $\rho(s_i)\le1$ gives a lower bound on the objective in
\eqref{eq:gaussian-parallel} by the corresponding objective in
\eqref{eq:gaussian-rdp}. For the reverse limiting bound,  fix a rate allocation $(r_i)$ attaining the minimum in \eqref{eq:gaussian-rdp}  and set $s_i=(R+C)/d$.
 Then every $\rho(s_i)$ tends to one as $C\to\infty$,
and continuity of the objective yields
\eqref{eq:gaussian-unlimited-benchmark}.
At zero  rate, all common-randomness budgets give
\[
 D_G(0,C,P)=V+\pos{\sqrt V-\sqrt P}^{\,2},
 \qquad 0\le C\le\infty.
\]

\bibliographystyle{IEEEtran}
\bibliography{references}
\end{document}